\documentclass[11pt]{amsart}  

\usepackage{euler}
\usepackage{amscd} 
\usepackage{mathrsfs}
\usepackage{amsmath,amsbsy,amsthm,amsfonts,amssymb}
\usepackage{hyperref}
\newtheorem{theorem}{Theorem}[section]
\newtheorem{lemma}[theorem]{Lemma}
\newtheorem{proposition}[theorem]{Proposition}
\newtheorem{corollary}[theorem]{Corollary} 
\theoremstyle{definition}

\newtheorem{remark}[theorem]{Remark} 
\numberwithin{equation}{section}

\def\FF{\mathscr{F}}
\def\psiht{\widehat{\psi}^{\hbar}}
\def\psih{\psi^{\hbar}}
\def\phih{\phi^{\hbar}}
\def\upsilonh{\Upsilon^{\hbar}}

\def\HE{\theta}
\def\vfi{\varphi}

\def\z*{\bar z}

\def\uno{\mathsf 1}
\def\H{\mathsf H}

\def\C{\mathcal C}

\def\dom{\text{\rm dom}}
\def\ran{\text{\rm ran}}

\def\FF{\mathscr F}

\def\RE{\mathbb R}
\def\CO{{\mathbb C}}
\def\o{$\bar{\text{\rm o}}$}

\def\ph*{\phi_\star}

\def\be{\begin{equation}}
\def\ee{\end{equation}}
\def\min{{\rm min}}
\def\max{{\rm max}}

\def\-{{\rm in}}
\def\+{{\rm ex}}
\newcommand{\sgn}{\operatorname{sgn}}

\def\Im{\operatorname{Im}}
\def\Re{\operatorname{Re}}
\def\sigmap{\breve{\sigma}}
\def\sigmaq{\sigma}
\def\sigmazero{\sigma_0}
\def\Szero{A}

\DeclareMathOperator*{\slim}{s-lim}

\title[The semi-classical limit with a delta-prime potential]{The semi-classical limit with a delta-prime potential}
\author{Claudio Cacciapuoti}
\address{DiSAT, Sezione di Matematica, Universit\`a dell'Insubria, via Valleggio 11, I-22100
Como, Italy}
\email{claudio.cacciapuoti@uninsubria.it}

\author{Davide Fermi}
\address{Classe di Scienze, Scuola Normale Superiore, Piazza dei Cavalieri 7, I-56126
Pisa, Italy}
\email{fermidavide@gmail.com}

\author{Andrea Posilicano}
\address{DiSAT, Sezione di Matematica, Universit\`a dell'Insubria, via Valleggio 11, I-22100
Como, Italy}
\email{andrea.posilicano@uninsubria.it}

\thanks{The authors acknowledge the support of the National Group of Mathematical Physics (GNFM-INdAM)}

\begin{document}

\begin{abstract} 
We consider the quantum evolution $e^{-i\frac{t}{\hbar}H_{\beta}}\psi_{\xi}^{\hbar}$ of a 
Gaussian coherent state $\psi_{\xi}^{\hbar}\in L^{2}(\RE)$ localized close to the classical state $\xi\equiv(q,p)\in\RE^{2}$, where $H_{\beta}$
denotes a self-adjoint realization of the formal Hamiltonian $-\frac{\hbar^{2}}{2m}\,\frac{d^{2}\,}{dx^{2}}\,{+}\,\beta\delta'_{0}$, with $\delta'_{0}$ the derivative of Dirac's delta distribution at $x=0$ and $\beta$ a real parameter. We show that in the semi-classical limit such a quantum evolution can be approximated (w.r.t. the $L^{2}(\RE)$-norm, uniformly for any $t\in\RE$ away from the collision time) by $e^{\frac{i}{\hbar}A_{t}}e^{it L_{B}}\phi^{\hbar}_{x}$, where $A_{t}=\frac{p^{2}t}{2m}$, $\phi_{x}^{\hbar}(\xi):=\psi^{\hbar}_{\xi}(x)$ and $L_{B}$ is a suitable self-adjoint extension of the restriction to $\C^{\infty}_{c}({\mathscr M}_{0})$, ${\mathscr M}_{0}:=\{(q,p)\!\in\!\RE^{2}\,|\,q\!\not=\!0\}$, of ($-i$ times) the generator of the free classical dynamics. While the operator $L_{B}$ here utilized is  similar to the one appearing in our previous work \cite{AMPA2020} regarding the semi-classical limit with a delta potential, in the present case the approximation gives a smaller error: it is of order $\hbar^{7/2-\lambda}$, $0\!<\!\lambda\!<\!1/2$, whereas it turns out to be of order $\hbar^{3/2-\lambda}$, $0\!<\!\lambda\!<\!3/2$, for the delta potential. We also provide similar approximation results for both the wave and scattering operators.
\end{abstract}

\maketitle

\begin{footnotesize}

\end{footnotesize}

\section{Introduction}

The close relation between coherent states and semi-classical analysis is well known and it goes back to the early days of Quantum Mechanics, see, e.g., \cite{CR} and references therein for a modern mathematical treatment. By Fourier transform, the classical-quantum correspondence is exact in the case of a free particle: 
defining, for any $\sigma\in\CO$,  the Gaussian coherent wave packet $\psi^{\hbar}_{\sigma,\xi}:\RE\to\CO$ centered at the classical phase space point $\xi\equiv(q,p)\in\RE^{2}$  by
\begin{equation}\label{initial}
\psi^{\hbar}_{\sigma,\xi}(x):=\frac{1}{(2\pi\hbar)^{1/4}\sqrt{\sigma}}\ \exp\left({-\,\frac{1}{4\hbar \sigmazero\sigma}\,(x-q)^{2}+\frac{i}{\hbar}\,p(x-q)}\right)
\end{equation}
and, for any $x\in\RE$, the phase space function  $\phi^{\hbar}_{\sigma,x}:\RE^{2}\to\CO$ by
\begin{equation}\label{phidef}
\phi^{\hbar}_{\sigma,x}(\xi):=\psi^{\hbar}_{\sigma,\xi}(x)\,,
\end{equation}
one has the relation\begin{equation}\label{free2}
\big(e^{-i\frac{t}{\hbar}\H_{0}}\,\psi^{\hbar}_{\sigmazero,\xi}\big)(x)=e^{\frac{i}{\hbar}\Szero_t}\,\big(e^{itL_{0}}
\phi^{\hbar}_{\sigma_{t},x}\big)(\xi)\,. 
\end{equation}
Here $$H_{0}: H^{2}(\RE)\subset L^{2}(\RE)\to L^2(\RE)\,,\quad H_{0}:=-\frac{\hbar^{2}}{2m}\,\frac{d^{2}\,}{dx^{2}}\,,$$ is the self-adjoint operator for a free quantum particle, 
$$
\sigma_{t}=\sigmazero+\frac{it}{2m\sigmazero}\,,\qquad A_{t}=\frac{p^{2}t}{2m}\,,
$$
and $e^{itL_{0}}$ is the realization in $L^{\infty}(\RE^2)$ of the strongly continuous (in $L^{2}(\RE^{2})$) group of evolution generated by the self-adjoint operator 
$$
L_{0}:=-\,i\,X_{0} \cdot \nabla\,,\qquad X_{0}(q,p):=\left(\frac{p}{m},0\right)\,,
$$ 
corresponding to the Hamiltonian vector field of a free classical particle, i.e.,  $$e^{itL_{0}}f(q,p) = f(q+\frac{p}{m}\,t,p)\,.
$$ 
Such an exact quantum-classical correspondence still holds for quadratic Hamiltonians and, for more general regular (at least $C^{2}$) potentials, an approximate relation, up to an error of order $\hbar^{1/2-\lambda}$, $0<\lambda<1/2$, is valid (see, e.g., \cite{Hag1}). In a previous paper, see \cite{AMPA2020}, we considered the semiclassical limit for a potential which is far from being a regular function, i.e. the case of  the Dirac delta distribution. Here we consider a still more singular case, that is we consider the case where the potential is given by the (distributional) derivative $\delta_{0}'$ of Dirac's delta.  Similarly to the case with a delta potential, the self-adjoint realization $H_{\beta}$, $\beta\in(\RE\backslash\{0\})\cup{\infty}$, of the formal Hamiltonian  $-\frac{\hbar^{2}}{2m}\,\frac{d^{2}\,}{dx^{2}}\,{+}\,\beta\delta'_{0}$  is described as a self-adjoint extension of the symmetric operator given by the restriction of the free Hamiltonian $\H_{0}$ to the set $\C^{\infty}_{c}(\RE\backslash\{0\})$ (see, e.g.,  \cite[Ch. I.4]{AGHH} and Section \ref{ss:deltaprime} below for more details). This fact, together with the free case relation \eqref{free2},
suggest (as in the case examined in \cite{AMPA2020}) how to proceede:
since $\H_{\beta}$ is a self-adjoint extension of the symmetric operator $\H^{\circ}_{0}:=\H_{0}\!\!\upharpoonright\!\C^{\infty}_{c}(\RE\backslash\{0\})$, one could try to approximate $e^{-i\frac{t}{\hbar}\H_{\beta}}\,\psi^{\hbar}_{\sigmazero,\xi}$ by replacing $L_{0}$ with $L_{B}$, a suitable self-adjoint extension   of $L^{\!\circ}_{0}:=L_{0}\!\!\upharpoonright\! \C^{\infty}_{c}({\mathscr M}_{0})$, ${\mathscr M}_{0}:=\RE^{2} \,\backslash\, \{(0,p) \,|\, p\!\in\!\RE\}$, and transforming $\phi^{\hbar}_{\sigma_{t},x}$ using the realization in $L^{\infty}(\RE^2)$, if any, of $e^{it L_{B}}$. 
Following the same reasonings as in \cite[Sec. 2]{AMPA2020}, in Section \ref{s:2} we provide the construction of $L_{B}$. In this introduction we content ourselves with giving the corresponding unitary group of evolution: 
for any $f\in L^{2}(\RE^{2})$ one has
\begin{equation}\label{Lbeta-group}
\big(e^{itL_{B}}f\big)(q,p)=\big(e^{itL_{0}}f\big)(q,p)-\frac{\HE (-tqp)\,\HE \!\left(\frac{|pt|}{m}-|q|\right)}{1-\sgn(t)\,\frac{2i|p|}{mB(p)}}\,\big(e^{itL_{0}}f_{odd}\big)(q,p)\,;
\end{equation}
here $\HE$ denotes the Heaviside function (namely, $\HE(x) = 1$ for $x > 0$ and $\HE(x) = 0$ for $x < 0$), $B(p)=bp^{2}$, $b\in(\RE\backslash\{0\})\cup\{\infty\}$, and $f_{odd}(\xi):=f(\xi)-f(-\xi)$. This shows that $e^{itL_{B}}$ is a group of evolution in $L^{\infty}(\RE^2)$. Notice that the case $b=\infty$ corresponds to complete reflection due to the infinite barrier at the origin, while the case $b\in\RE\backslash\{0\}$ allows transmission ($b=0$ gives the free generator $L_{0}$) thus introducing ``extra'' classical paths going beyond the singularity. 
\par  In Subsection \ref{ss:3.1} we prove the following 

\begin{theorem}\label{t:1} Let ${B}(p) := -\, (2  \beta/\hbar^3)\,p^2$. Then, there exists a constant $C>0$ such that, for any $\eta \in (0,1)$, for any $t\in\RE$ and for any $\xi\equiv(q,p)\in\RE^{2}$ with $qp\not=0$, there holds
\begin{equation}\label{t1}
\begin{aligned}
&\left\|e^{-i\frac{t}{\hbar}\H_{\beta}}\,\psi^{\hbar}_{\sigmazero,\xi}-
e^{\frac{i}{\hbar}\Szero_{t}} \big(e^{itL_{B}} \phi^{\hbar}_{\sigma_{t},(\cdot)}\big)(\xi)\right\|_{L^{2}(\RE)}\\
& \leq  C \Bigg[ {\eta \over 1- \eta}\, \Big({\hbar^3 \over m |\beta p|}\Big) + e^{-\,\eta^2 {\sigma_0^2 p^2 \over 2\hbar}} + e^{- {q^2 \over 4\hbar \sigma_0^2}} + e^{-{\sigma_0^2 p^2 \over \hbar}} \\
& \hspace{1cm} + \left({\hbar^{5} \sigma_0^2 \over m^2\,\beta^2}\right)^{\!\!1/4} e^{{\hbar^5 \sigma_0^2 \over m^2 \beta^2}} \left(e^{- {\sigma_0^2 p^2 \over \hbar}} + e^{- \frac{q^2}{4\hbar \sigma_0^2}}\right) + e^{- \frac{q_t^2}{4\hbar|\sigma_t|^2}}\!\Bigg] \,. 
\end{aligned}
\end{equation}
\end{theorem}

Thus, whenever $t$ is not too close to the collision time $t_{coll}(\xi):=-\frac{mq}{p}$, $\xi\equiv(q,p)$  (look at  the last term in the above estimate), our approximation provides the following result (see Subsection \ref{proof-coroll}):

\begin{corollary}\label{c:1}
Let ${B}(p) := -\, (2  \beta/\hbar^3)\,p^2$.  Then, for any $0 < \lambda < 1/2$ there exit constants $0 < h_{*} < 1$ and $C_{*}, c_0 > 0$ such that
\begin{equation*}
\underline h := \max\left\{\frac{\hbar \sigmazero^2}{q^2}\,,\,{\hbar \over \sigmazero^2 p^2}\,,\,{\hbar \over (m |\beta p|)^{1/3}} \right\} < h_{*}
\end{equation*}
implies 
$$\left\|e^{-i\frac{t}{\hbar}\H_{\beta}}\,\psi^{\hbar}_{\sigmazero,\xi}-
e^{\frac{i}{\hbar}\Szero_{t}} \big(e^{itL_{B}} \phi^{\hbar}_{\sigma_{t},(\cdot)}\big)(\xi)\right\|_{L^{2}(\RE)}\le C_{*}\,\underline h^{\frac{7}{2}- \lambda}\,,
$$ 
for any $t\in\RE$, $\xi\equiv(q,p)\in\RE^{2}$ with $qp\not=0$, such that
\begin{equation}\label{ttcoll}
\big|t-t_{coll}(\xi)\big|\ge c_{0}\,|t_{coll}(\xi)|\,\sqrt{\Big(\frac{7}{2}-\lambda\Big)\,\underline h\,|\!\ln \underline h|}\;.
\end{equation}
\end{corollary}
Moreover, the constraint  $t\not=t_{coll}$ does not affect the semi-classical approximation for large times. Indeed, see Theorem \ref{thm: Ompm} below, we can handle the approximation of the wave operators: denoting with $\Omega^{\pm}_{\beta}$ the wave  operators defined, as usual, by the limits in 
$L^{2}(\RE^{2})$
$$
\Omega^{\pm}_{\beta}f:=\lim_{t\to\pm\infty}e^{i\frac{t}{\hbar}\H_{\beta}}e^{-i\frac{t}{\hbar}\H_{0}}f
$$
and by $W^{\pm}_{B}$ the corresponding classical objects (compare with \cite[Def. 3.4.4]{Th}, see also \cite[Rem. 2.7]{AMPA2020})
\begin{equation*}%\label{wo-cl}
W^{\pm}_{B}f:=\lim_{t\to\pm\infty}e^{i tL_{0}}e^{ - i  t L_{B}}f
\end{equation*}
(here the limits hold both pointwise in $\RE^{2}$ and, if $f = \psih_{\sigma,\xi}$ is a coherent state of the form \eqref{initial}, in $L^{2}(\RE,dx)$, see Proposition \ref{prop-wave-cl} below), one has the following (see Subsection \ref{convWops} for the proof) 
\begin{theorem}\label{thm: Ompm}
Let ${B}(p) := -\,(2  \beta/\hbar^3)\,p^2$. Then, for any $\psi^{\hbar}_{\sigmazero,\xi} \in L^2(\RE)$ of the form \eqref{psi0} with $qp\not=0$ and for any $\eta \in (0,1)$, there exists a constant $C>0$ such that
\begin{align}\label{t2_1}
\left\|\Omega_\beta^{\pm}\, \psi^{\hbar}_{\sigmazero,\xi} - \big(W_{{B}}^{\pm} \phi^{\hbar}_{\sigmazero,\,(\cdot)} \big)(\xi) \right\|_{L^2(\RE)} \! \leq C \left[ {\eta \over (1- \eta)}\, \Big({\hbar^3 \over m |\beta p|}\Big) + e^{-\,\eta^2 {\sigmazero^2 p^2 \over 2\hbar}}	+ e^{- \frac{q^2}{4\hbar \sigmazero^2}} + e^{- {\sigmazero^2 p^2 \over \hbar}} \right] .
\end{align}
Similarly, for the scattering operators $S_{\beta}:=(\Omega_{\beta}^{+})^{*}\Omega_{\beta}^{-}$ and $S_{B}^{cl}:=(W_{B}^{+})^{*}W_{B}^{-}$ there holds
\begin{align}\label{t2_2}
\left\|S_\beta\, \psi^{\hbar}_{\sigmazero,\xi} - \big(S^{cl}_{B} \phi^{\hbar}_{\sigmazero,\,(\cdot)} \big)(\xi) \right\|_{L^2(\RE)} \! & \leq 
C \left[ {\eta \over (1- \eta)}\, \Big({\hbar^3 \over m |\beta p|}\Big) + e^{-\,\eta^2 {\sigmazero^2 p^2 \over 2\hbar}}	+ e^{- \frac{q^2}{4\hbar \sigmazero^2}} + e^{- {\sigmazero^2 p^2 \over \hbar}} \right. \\
& \hspace{2cm}
\left. +\, \left({\hbar^{5} \sigma_0^2 \over m^2 \beta^2}\right)^{\!\!1/4} e^{{\hbar^5 \sigma_0^2 \over m^2 \beta^2}} \left( e^{- {\sigmazero^2 p^2 \over \hbar}} + e^{- \frac{q^2}{4\hbar \sigmazero^2}} \right) \right] . \nonumber
\end{align}
\end{theorem}

\begin{corollary}\label{cor: Ompm}
For any $0 < \lambda < 1/2$ there exit constants $0 < h_{*} < 1$ and $C_{*} > 0$ such that
%\begin{equation}
%\underline h := \max\left\{\frac{\hbar \sigmazero^2}{q^2}\,,\,{\hbar\, \over \sigma_0^2 p^2}\,,\,{\hbar \over (m |\beta p|)^{1/3}} \right\} < h_{*}\,,
%\end{equation}
%implies
\begin{gather*}
\left\|\Omega_\beta^{\pm}\, \psi^{\hbar}_{\sigma_0,\xi} - \big(W_{{B}}^{\pm} \phi^{\hbar}_{\sigma_0,\,(\cdot)} \big)(\xi)  \right\|_{L^2(\RE)} \leq C_{*} \;\underline{h}^{7/2 -\lambda}\, , \\
\left\|S_\beta\, \psi^{\hbar}_{\sigma_0,\xi} - \big(S^{cl}_{B} \phi^{\hbar}_{\sigma_0,\,(\cdot)} \big)(\xi)  \right\|_{L^2(\RE)} \leq C_{*} \;\underline{h}^{7/2 -\lambda}\, .
\end{gather*}
\end{corollary}
\begin{remark} Theorems \ref{t:1} and \ref{thm: Ompm} (and the relative Corollaries) parallel the analogous ones in \cite{AMPA2020} (see Theorems 1.1 and 1.3 therein) which provide semi-classical approximations for the quantum evolutions, wave and scattering operators for the operator $H_{\alpha}$ providing a self-adjoint realization of the formal Hamiltonian $-\frac{\hbar^{2}}{2m}\,\frac{d^{2}\,}{dx^{2}}\,{+}\,\alpha\,\delta_{0}$ (see, e.g., \cite[Ch. I.3]{AGHH}). The classical approximating self-adjoint operator $L_{\beta}$ used there (see \cite[Sec. 2]{AMPA2020}) is not too much different form the operator $L_{B}$ used here: the group of evolution generated by $L_{\beta}$ is given by (compare with \eqref{Lbeta-group})
$$
\big(e^{itL_{\beta}}f\big)(q,p)=\big(e^{itL_{0}}f\big)(q,p)-\frac{\HE (-tqp)\,\HE \!\left(\frac{|pt|}{m}-|q|\right)}{1-\sgn(t)\,\frac{2i|p|}{m\beta}}\,\big(e^{itL_{0}}f_{ev}\big)(q,p)\,,
$$
where (see \cite[Prop. 2.4]{AMPA2020}) $f_{ev}(\xi):=f(\xi)+f(-\xi)$. However, the mentioned results in \cite{AMPA2020} give an error of different order: $\underline{\underline h}^{3/2-\lambda}$ for $0\!<\!\lambda\!<\!3/2$, where 
$$
\underline{\underline h}:=\max\left\{\frac{\hbar \sigmazero^2}{q^2}\,,\,{\hbar\, \over \sigma_0^2 p^2}\,,\,{\hbar \over (m |\alpha|\sigma_{0})^{2/3}} \right\}\,.
$$
By techniques similar to the ones used here and in \cite{AMPA2020}, analogous semiclassical estimates can also be obtained for the case of a quantum evolution on graphs, see \cite{AAMP}. 
\end{remark}
%Unless otherwise stated, we refer to the principal determination for the argument of complex numbers, \emph{i.e.}, $\arg : \CO \setminus [0,+\infty) \to (0,2\pi)$. This is such that $\Im \sqrt{z} > 0$ for any $z \in \setminus [0,+\infty)$.

\section{Singular perturbations of the free classical dynamics\label{s:2}}

By the same kind of reasonings as in \cite[Sec. 2]{AMPA2020}, in this section we introduce a suitable self-adjoint extension of the restriction to functions vanishing on the line $\{(q,p)\in\RE^{2}: q=0\}$ of the self-adjoint operator 
$-i\frac{p}{m}\partial_{q}$.
At variance with the self-adjoint operator $L_{\beta}\equiv L_{\Pi,\beta}$ provided in \cite{AMPA2020}, the operator $L_{B}\equiv L_{\Pi',B}$ here defined corresponds to different choices of both the extension parameters: in \cite{AMPA2020} $\Pi$ is the projector onto the subspace of ($p$-dependent) even functions and the operator $\beta$ identifies with the multiplication by the constant $\beta\in(\RE\backslash\{0\})\cup\{\infty\}$, while here we use the projection $\Pi'$ onto the subspace of odd functions and the operator $B$ identifies with the multiplication by the function $B(p):=b p^{2}$, $b\in (\RE\backslash\{0\})\cup\{\infty\}$. Notwithstanding such differences, the proofs of the results presented in this section follow almost verbatim the ones of the corresponding results in \cite[Sec. 2]{AMPA2020} and therefore are not reproduced here.    
\par
Let $X_{0}(q,p)=(p/m,0)$ be the Hamiltonian vector field  of a classical free particle in $\RE$ and let
$$
L_{0}:\dom(L_{0})\subseteq L^{2}(\RE^{2})\to L^{2}(\RE^{2})\,,\qquad L_{0}f  = -\,i\,X_{0}\!\cdot\!\nabla f\,, 
$$
defined on the maximal domain $\dom( L_{0}) := \{ f \in L^2(\RE^2) \;|\; X_{0}\!\cdot\!\nabla f\in L^2(\RE^2) \} $, be the corresponding self-adjoint operator in $L^{2}(\RE^{2})$; one has $\sigma(L_{0})=\sigma_{ac}(L_{0})=\RE$. 
\par 
The linear map $(\gamma f)(p):=f(0,p)$ extends to a bounded operator $\gamma :\dom(L_{0}) \to  L^2(\RE, |p|\,dp)$ (here $\dom(L_{0}) \subset L^2(\RE)$ is endowed with the graph norm) such that 
$\ker(\gamma)$ is dense in $L^{2}(\RE^{2})$ (see \cite[Lem. 2.1]{AMPA2020}).\par
Denoting by $R^{0}_z := (L_{0} -z)^{-1}$ for $z\in\CO\backslash\RE$ the resolvent of $L_{0}$, one gets
$$
(R^{0}_z f)(q,p) = \int_{\RE}\!  dq'\, g_z(q - q',p)\, f(q',p)\,,
$$
where 
$$
g_z(q,p) = \HE(q\,p \Im z)\;\sgn(\Im z)\;\frac{i\,m}{|p|}\,{e^{i m z q/p}}
$$
(recall that  $\HE$ indicates the Heaviside step function). For any  $z\in\CO\backslash\RE$, we define the bounded linear map
$$
G_{z}:L^{2}(\RE,|p|^{-1}dp)\to L^{2}(\RE^{2})\,,\qquad 
(G_{z}\phi)(q,p) :=((\gamma\, R^0_{\bar z})^{*}\phi)(q,p) \equiv g_{z}(q,p)\,\phi(p)\,;
$$
here $L^{2}(\RE,|p|^{-1}dp)$ and $L^{2}(\RE,|p|dp)$ are considered as a dual couple with respect to the duality induced by the scalar product in $L^{2}(\RE)$. 
\par
For any $b\in(\RE\backslash\{0\})\cup\{\infty\}$ we define the function 
\be\label{b}
B:\RE\to \RE\cup\{\infty\}\,,\qquad B(p):=\begin{cases}b p^{2}& b\in \RE\backslash\{0\}\\ \infty& b=\infty
\end{cases}
\ee
and then, for any $z\in\CO\backslash\RE$, we define the bounded linear map 
$$
\Lambda^{B}_{z}:L^{2}(\RE,|p|dp)\to L^{2}(\RE,|p|^{-1}dp) \,,\quad
(\Lambda^{B}_{z}\phi)(p):=\frac{\phi(p)}{\frac1{B(p)} - \sgn(\Im z)\,\frac{i\,m}{2\,|p|}}
$$
(here we set $\frac{1}{\infty}:=0$). Finally, we introduce the projector on odd functions (here either $\rho(p)=|p|$ or $\rho(p)=|p|^{-1}$)
$$
\Pi' : L^2(\RE,\rho dp) \to L^2(\RE,\rho dp)\,, \qquad (\Pi' f)(p) := {1 \over 2}\,\big( f(p) - f(-p) \big)
$$
and notice that
$$ 
%\Pi'\, M^{B}_{z} = M^{B}_{z}\, \Pi' \,, \qquad 
\Pi' \,\Lambda^{B}_{z} = \Lambda^{B}_{z}\, \Pi' \,.
$$
Then, by \cite[Thm 2.1]{P2001} here employed with $\tau:=\Pi'\gamma$, we obtain the following (compare with \cite[Thm. 2.2]{AMPA2020})
\begin{theorem}\label{resolvent}
For any $b\in (\RE \backslash \{0\}) \cup \{\infty\}$, and $B$ defined as in \eqref{b}, the linear bounded operator
\be\label{krein}
R^{B}_{z} := R_{z}^{0} + G_{z}\, \Pi'\, \Lambda^{B}_{z}\, \Pi'\, G^{*}_{\bar z}\quad \mbox{with\, $z\in\CO\backslash\RE$}\,,
\ee
is the resolvent of a self-adjoint extension $L_{B}$ of the densely defined, closed symmetric operator $L_{0}\!\upharpoonright\!\ker(\gamma)$. It acts on its domain 
$$
\dom(L_{B}):= \big\{f\in L^{2}(\RE^{2}) \;\big|\; f=f_{z}+G_{z}\,\Lambda^{B}_{z}\,\Pi'\,\gamma f_{z}\,,\; f_{z}\in\dom(L_{0})\big\}\,,
$$
by
$$
(L_{B}-z)f=(L_{0}-z)f_{z}\,.
$$
\end{theorem}

\begin{remark}\label{remBC}
Notice that the functions $f=f_{z}+G_{z}\phi$, $\phi\in\ran(\Pi')$, belonging to $\dom( L_{B})$ fulfill the boundary condition
$$
\phi(p)=B(p)(\Pi'\,\widehat{\gamma}\,f )(p)\,.
$$
where $\widehat{\gamma}$ is the extension of the trace map $\gamma$ defined as
$$
(\widehat{\gamma}\, f)(p) := {1 \over 2}\, \big(f(0_{+},p)\, + f(0_{-},p) \big) \,.
$$
Moreover, on account of the basic identity $(-i\,X_{0} \cdot \nabla- z) G_z \phi = \phi\,\delta_{\Sigma_{0}}$, 
where $\phi\delta_{\Sigma_{0}}$ is the distribution supported on the line ${\Sigma_{0}}=\{(q,p)\!\in\!\RE^{2}\,|\, q=0\}$ defined by 
$$
\big(\phi\,\delta_{\Sigma_{0}}\big)(\varphi):=\int_{\Sigma_{0}}\!\!dp\;\phi(p)\,\varphi(0,p)\,,\quad\; \mbox{for any\, $\varphi\in \C^{\infty}_{c}(\RE^{2})$}\,,
$$
from Theorem \ref{resolvent} one can readily infer that
\begin{equation*}
L_{B} f = -i\,X_{0} \cdot \nabla f - \phi\,\delta_{\Sigma_{0}} \equiv 
-i\,X_{0}\cdot \nabla f - (B\Pi'\,\widehat{\gamma}\,f )\,\delta_{\Sigma_{0}}\,.
\end{equation*}
\end{remark}
By functional calculus and by \eqref{krein}, the action of the unitary group $e^{- i t L_{B}}$ ($t \in \RE$) describing the dynamics induced by $L_{B}$ can be explicitly characterized (the proof coincides with the one for \cite[Prop. 2.4]{AMPA2020} by noticing that all the integrals appearing there regard the $q$-variable only):
\begin{proposition}\label{exp-cl} Let $L_{B}$ be as in Theorem \ref{resolvent} and $f \in L^2(\RE^2)$. Then 
\begin{equation}\label{eq: timeevcl}
\big(e^{- i t L_{B}} f\big)(q,p) = \big(e^{- i t L_{0}} f\big)(q,p) \,-\, \frac{\HE(t\,q\,p)\; \HE\big({|p\,t| \over m} - |q|\big)}{1 +\sgn(t)\,{ 2i\,|p| \over mB(p)}}\, \big(e^{- i t L_{0}} f_{odd}\big)(q,p)\;,
\end{equation}
where $f_{odd}(q,p):=f(q,p)-f(-q,-p)$ and $e^{- i t L_{0}}$ denotes the free unitary group
\begin{equation*}%\label{ExpFree}
\big(e^{- i t L_{0}} f\big)(q,p) = f\Big(q - {p t \over m} , p\Big)\,.
\end{equation*}
\end{proposition}
\begin{remark} Formula \eqref{eq: timeevcl} shows that $e^{-itL_{B}}$ defines a group of evolution in 
 $L^{\infty}(\RE^2)$.
\end{remark}
\begin{remark}\label{RemDir}
Notice that while the free operator $e^{- i t L_{0}}$ maps real-valued functions into real-valued functions, the same is not true for  $e^{- i t L_{B}}$, unless $b = \infty$, which corresponds to a complete reflection. In this particular case, Eq. \eqref{eq: timeevcl} reduces to
\begin{align*}
&\big(e^{-it L_{B}}f\big)(q,p) \\
& \overset{\text{$(b=\infty)$}}{=} \left[1 - \HE(t\,q\,p)\;\HE\!\left({|p\,t| \over m} - |q|\right)\!\right] f\Big(q - {pt \over m},p\Big)
  +\,\HE(t\,q\,p)\; \HE\!\left({|p\,t| \over m} - |q|\right) f\Big(\!-q + {pt \over m},-\,p\Big)\,.
\end{align*}
%This also justifies our introduction of the projector $\Pi$ defined in \eqref{proj}: it leads to a family of self-adjoint extensions containing the generator of the dynamics corresponding to complete reflection.
\end{remark}
Defining the classical wave operators by 
\be\label{w-class}
W^{\pm}_{B}f:=\lim_{t\to\pm\infty}e^{   i tL_{0}}e^{ - i  t  L_{B}}f\,,
\ee
%and the corresponding classical scattering operator by
%\[
%S^{cl}_{B}  := (W_B^{+})^{*}\, W_B^{-}
%\]
one then has the following (compare with \cite[Prop. 2.8]{AMPA2020})
\begin{proposition}\label{prop-wave-cl} The limits in \eqref{w-class} exist pointwise for any $\xi\equiv(q,p)\!\in\! \RE^2$ with $qp\not=0$ and in $L^{2}(\RE^2)$ for any $f\in L^{2}(\RE^2)$: 
\begin{equation}\label{eq: waveopcl}
\big(W_{B}^{\pm} f\big)(q,p) 
= f(q,p) - {\HE(\mp q p) \over 1 \pm {2 i\, |p| \over mB(p)}}\,f_{odd}(q,p)\,.
\end{equation}
Furthermore, the classical scattering operator $S^{cl}_{B}  := (W_B^{+})^{*}\, W_B^{-}$ is given by 
\begin{equation}\label{eq: scattopcl}
\big( S^{cl}_{B} f\big)(q,p) = f(q,p) - {f_{odd}(q,p) \over 1 - {2 i\, |p| \over mB(p)}}\;.
\end{equation}
\end{proposition}

\begin{remark}\label{r:Wpm-unitary} On account of Eq. \eqref{eq: waveopcl}, it is easy to check that 
\[
W_{B}^{+} W_{B}^{-} f = W_{B}^{-} W_{B}^{+} f\,.
\]
Moreover, from the identity 
\[\frac1{1 + {2i\,|p| \over mB(p)}} + \frac1{1 - {2i\,|p| \over mB(p)}}  =\frac{2}{ \big|1 + {2i\,|p| \over mB(p)}\big|^{2}}
\] 
 and a straightforward calculation it follows that 
\[
W_{B}^{\pm} (W_{B}^{\pm})^{*} f = (W_{B}^{\pm})^{*}\, W_{B}^{\pm} f = f\,.
\]
Hence, in particular, $ S^{cl}_{B} W_{B}^{+} f =  W_{B}^- f$.
\end{remark}

\begin{remark} By arguments similar to those used in the proof of Proposition \ref{prop-wave-cl}, one gets that the limits 
$$
\breve W^{\pm}_{B}f:=\lim_{t\to\pm\infty}e^{   i tL_{B}}e^{ - i  t L_{0}}f
$$ 
%% exist pointwise in $\RE^2$, in $L^{2}(\RE^2)$ for any $f\in L^{2}(\RE^2)$ and
exist in $L^{2}(\RE^2)$ for any $f\in L^{2}(\RE^2)$ and
$$
\big(\breve W^{\pm}_{B}f\big)(p,q) = f(q, p) - \; \frac{\HE(\mp q\,p)}{1 \mp {2i\,|p| \over mB(p)}}\  f_{odd}(q, p)\,.
$$
Therefore, by \cite[Ch. X, Thm. 3.5]{Kato}, both $W^{\pm}_{B}$ and $\breve W^{\pm}_{B}$ are complete, and the absolutely continuous part of $L_{B}$ is unitarily equivalent to the absolutely continuous part of $L_{0}$, i.e., to  $L_{0}$ itself; thus 
$$
\sigma_{ac}(L_{B})=\sigma_{ac}(L_{0})=\RE
$$ 
and $L_{B}$ is unitarily equivalent to $L_{0}$.
\end{remark}

\section{The quantum Hamiltonian with a delta-prime potential}\label{ss:deltaprime}
Here we recall the definition and main properties of the operator $H_{\beta}$, $\beta\in \RE\cup{\infty}$, defined as a self-adjoint extension of the symmetric operator given by the restriction of the free Hamiltonian 
$$
\H_{0}:H^{2}(\RE)\subset L^{2}(\RE)\to L^{2}(\RE)\,,\quad \H_{0}:= - \,{\hbar^2 \over 2m}\,\frac{d^{2}\ }{dx^{2}}\,,
$$
to the set $\{\psi\in H^{2}(\RE):\psi(0)=0\}$, where $H^{2}(\RE)$ denotes the usual Sobolev space of order two, namely $H^{2}(\RE):=\{\psi\!\in\! L^{2}(\RE)\,|\,\psi''\!\in\! L^{2}(\RE)\}$. In more detail, one has (see \cite[Thms. 4.2 and 4.3]{AGHH})
\begin{gather}
\dom (H_{\beta}) := \Big\{ \psi \!\in\! H^2(\RE \backslash \{0\}) \,\Big|\, \psi'(0^{+}) = \psi'(0^{-})=\psi'(0), \ \psi(0^{+}) - \psi(0^{-}) = {2m \beta \over \hbar^2}\,\psi'(0) \Big\}\,, \nonumber %\label{eq: domHbeta} 
\\
H_{\beta} \psi=  - \,{\hbar^2\over 2m}\, \psi''+\beta\,\psi'(0)\,\delta_{0}'\,. \nonumber%\label{eq: Hbeta}
\end{gather}
Moreover,
$$
\sigma_{ac}(H_{\beta})=[0,+\infty)\,,\qquad\sigma_{sc}(H_{\beta})=\varnothing\,,
$$
\begin{equation*}
\sigma_{p}(H_{\beta}) = \left\{\!\begin{array}{ll}
\varnothing & \mbox{if $\beta \geqslant 0$ or $\beta = \infty$\,,} \\
\displaystyle{\Big\{\lambda_{\beta} \equiv -\, {\hbar^6 \over 2 m^3\,\beta^2 }\Big\} } & \mbox{if $\beta < 0$}\,.
\end{array}\right.
\end{equation*}
The normalized eigenfunction associated to the negative eigenvalue for $\beta < 0$ reads
\begin{equation*}
\vfi_{\beta}(x) = {\hbar \over \sqrt{m\,|\beta|}} \,\sgn(x)\,e^{-{\hbar^2 \over m\,|\beta|}\,|x|}\,.
\end{equation*}

\begin{remark} The possible eigenvalue approaches the absolutely continuous spectrum from below in the semiclassical limit, \emph{i.e.}, $\lambda_{\beta} \to 0^{-}$ for $\hbar \to 0^{+}$; correspondingly, the associated eigenfunction vanishes almost everywhere, namely $\|\vfi_{\beta} \|_{L^{\infty}(\RE)} \to 0$. This marks a noteworthy difference with respect to the case of a delta potential discussed in \cite{AMPA2020}, where the possible eigenvalue moves away from the absolutely continuous part of the spectrum and the associated eigenfunction becomes sharply peaked at one point for $\hbar \to 0^{+}$. As a consequence, many of the arguments employed in \cite{AMPA2020} cannot be implemented in the present setting.
\end{remark}
A complete set of generalized eigenfunctions associated to the absolutely continuous part of the spectrum is given by (compare with \cite[Eq. (4.23)]{AGHH})
\begin{gather}\label{eq: eigenfdeltap}
\vfi^{\pm}_{k}(x) := {e^{i \,k\,x} \over \sqrt{2\pi}} + R_{\pm}(k)\,\sgn(x)\,{e^{\mp i \,|k|\,|x|} \over \sqrt{2\pi}}  \qquad (k \in \RE \backslash \{0\}) \,, \\
R_{\pm}(k) := {{i m\,\beta\,k \over \hbar^2} \over 1 \pm  {i m \,\beta\, |k| \over \hbar^2}} = \pm\, {k \over |k| \mp {i \hbar^2 \over m \beta}} \,. \label{eq: Rpm}
%% = \pm\,\sgn(k)\, {1 \over 1 \pm {i \hbar^2 \over m \beta\,|k|}}
\end{gather}
Notice that
\begin{equation}\label{basic}
R_{+}(k) - R_{-}(k) = 2\,\sgn(k)\,|R_{+}(k)|^2\,.
\end{equation}
For any $\beta \in \RE$, taking into account the above spectral decomposition of $\H_{\beta}$, let us consider the bounded operators 
\begin{gather}
\FF : L^2(\RE) \to L^2(\RE)\,, \qquad 
(\FF\, \psi)(k) := \int_{\RE}\! dx\; {e^{-i k x} \over \sqrt{2\pi}}\; \psi(x) \,, \nonumber \\
\FF_{\pm} : L^2(\RE) \to L^2(\RE)\,, \qquad 
(\FF_{\pm}\, \psi)(k) := \int_{\RE}\! dx\; \overline{\vfi^\pm_k(x)}\; \psi(x) \,. \label{eq: FFpm}
\end{gather}
Correspondingly, we introduce the orthogonal projectors
\begin{gather}
P_{ac} : L^2(\RE) \to L^2(\RE)\,, \qquad (P_{ac} \psi)(x) := \int_{\RE}\! dk\;\vfi_k^+(x)\; (\FF_{+}\, \psi)(k) \,, \label{defPac} \\
P_{\beta} : L^2(\RE) \to L^2(\RE)\,, \qquad (P_{\beta} \psi)(x) := \theta(-\beta)\, \vfi_{\beta}(x) \int_{\RE}\! dy\; \vfi_{\beta}(y)\; \psi(y) \,; \label{defPpp}
\end{gather}
these are such that
\begin{equation}
P_{ac} + P_{\beta} = \uno\,. \label{orto}
\end{equation}
Eqs. \eqref{defPac} and \eqref{defPpp} reduce to
\begin{equation*}
P_{ac} = \uno\,, \quad P_{\beta} = 0 \qquad \mbox{for $\beta \geqslant 0$}\,. %\label{Pacuno}
\end{equation*}

For any $\beta \!\in\! \RE$ the time evolution of any state $\psi \!\in\! L^2(\RE)$ induced by the unitary group $e^{- i {t \over \hbar} \H_\beta}$ can be characterized as
\begin{equation}\label{qevol}
\big(e^{-i \frac{t}{\hbar} \H_\beta}\, \psi \big)(x) = \int_{\RE}\! dk\; e^{-i {t \over \hbar}\frac{\hbar^2 k^2}{2m}}\; \vfi_{k}^+(x)\, (\FF_{+}\, \psi)(k) + e^{-i {t \over \hbar}\, \lambda_\beta}\, (P_{\beta}\psi)(x)\,.
\end{equation}
In the definition of $P_{ac}$ and in Eq. \eqref{qevol}, one could equivalently use the generalized eigenfunctions $\vfi_k^-$ and the bounded operator $\FF_{-}$, respectively in place of $\vfi_k^+$ and $\FF_{+}$.  
\par
Since $(H_{\beta}-z)^{-1}-(H_{0}-z)^{-1}$, $z\in\CO\backslash\RE$, is a rank-one operator (see \cite[Thm. 4.1]{AGHH})  existence and completeness of the wave operators 
$$
\Omega_\beta^\pm:=\slim_{t\to\pm\infty}e^{i \frac{t}{\hbar} \H_\beta}e^{-i \frac{t}{\hbar} \H_0}
$$ follows from \cite[p. 550, Thm. 4.12]{Kato}; in particular, 
%let us stress that $\Omega_\beta^\pm$ are unitary on the absolutely continuous subspace for $\H_{\beta}$, namely, on  $\ran (P_{ac})$ with $P_{ac}$  defined according to Eq. \eqref{defPac}; more precisely, there holds 
$(\Omega_\beta^{\pm})^{*} {\Omega}_\beta^{\pm} = P_{ac}$. The corresponding scattering operator
is defined, as usual, by $$S_{\beta}:=(\Omega_{\beta}^{+})^{*}\Omega_{\beta}^{-}\,.
$$ 
Moreover, one has
\begin{equation}\label{st-rep}
\Omega_\beta^\pm = \FF_{\pm}^{*} \FF \,.
%\qquad \text{and} \qquad \Omega_\alpha^- = {\FF_{-}}^{*} \FF 
\end{equation}
%and 
%\[
%({\FF_{\pm}})^*\tilde \psi) (x) =  \int_{\RE}\! dk\;\vfi_k^\pm(x)\; \tilde \psi(k) \,.
%\]
Relation \eqref{st-rep} is well known in the case of perturbations by regular potentials and can also be proved, by essentially the same kind of proof, in the case of a singular perturbation (see the proof of \cite[Thm. 5.5]{JST}).

\section{Convergence of the dynamics}
We focus our attention on coherent states of the form
\begin{equation}\label{CohSt}
\psih(x) \equiv \psih(\sigmaq,\sigmap,q,p;x) = \frac{1}{(2\pi \hbar)^{1/4} \sqrt{\sigmaq}}\; e^{- \frac{\sigmap}{4\hbar\sigmaq} (x-q)^2 + i \frac{p}{\hbar} (x-q)} \qquad (x \in \RE)\,,
\end{equation}
where $(q,p) \in \RE^2$ and $\sigmaq,\sigmap \in \CO$ are such that
\begin{equation}\label{sxsp1}
\Re \sigmaq > 0\,, \qquad \Re \sigmap > 0\,, \qquad
\Re \!\big[\overline{\sigmaq}\,\sigmap\big] = 1\,.
\end{equation}
The Fourier transform  with respect to $x$ of any state $\psih$ of the form \eqref{CohSt} reads
\begin{align}
\big(\FF\psih\big)(k)\equiv \psiht (\sigmaq,\sigmap,q,p;k) 
%& = {1 \over \sqrt{2\pi}} \int_{\RE}\!\!dx\;e^{-i k x}\,\psih(\sigmaq,\sigmap,q,p;x) \nonumber \\
& = {1 \over \sqrt{\sigmap}} \left({2 \hbar \over \pi}\right)^{\!\!1/4} e^{-{\hbar \sigmaq \over \sigmap}(k- p/\hbar)^2 -  i k q}\,. \label{psih}
\end{align}

From now on we fix $\sigmazero>0$, $\sigmap=\sigmazero^{-1}$ and define, for any $\sigma\in\CO$ with $\Re \sigma = \sigmazero$, $\xi\equiv(q,p)\in\RE^{2}$, the state $\psi^{\hbar}_{\sigma,\xi}$ as in Eq. \eqref{initial}; notice that $\psi^{\hbar}_{\sigma,\xi} \equiv  \psih(\sigmaq,\sigmazero^{-1},q,p)$. 
%From now on we fix $\sigmazero>0$, $\sigmap=\sigmazero^{-1}$ and define, for any $\sigma\in\CO$ with $\Re \sigma = \sigmazero$, $\xi\equiv(q,p)\in\RE^{2}$, 
%\begin{equation}\label{initial}
%\psi^{\hbar}_{\sigma,\xi}:\RE\to\CO\,,\quad
%\psi^{\hbar}_{\sigma,\xi}(x):=\frac{1}{(2\pi\hbar)^{1/4}\sqrt{\sigma}}\ \exp\left({-\frac{1}{4\hbar \sigmazero\sigma}\,(x-q)^{2}+\frac{i}{\hbar}\,p(x-q)}\right).
%\end{equation}
%For any $\sigma\in\CO$ and $x\in\RE$, we put
%\begin{equation}\label{phidef}
%\phi^{\hbar}_{\sigma,x}:\RE^{2}\to\CO\,,\quad
%\phi^{\hbar}_{\sigma,x}(\xi):=\psi^{\hbar}_{\sigma,\xi}(x)\,.
%\end{equation}

In the sequel we analyze the time evolution, generated by the unitary group $ e^{-i \frac{t}{\hbar} \H_\beta}$ ($t \in \RE$), of an initial state of the form
\begin{equation}\label{psi0}
 \psih_{\sigmazero,\xi}(x)  = \psih(\sigmazero,\sigmazero^{-1},q,p;x)
 \qquad \big(\xi = (q,p)\big)\,.
\end{equation}

%For later reference, let us recall that
%\begin{equation}\label{free2}
%\big(e^{-i\frac{t}{\hbar}\H_{0}}\,\psi^{\hbar}_{\sigmazero,\xi}\big)(x)=e^{\frac{i}{\hbar}\Szero_t}\,\big(e^{itL_{0}}
%\phi^{\hbar}_{\sigma_{t},x}\big)(\xi)\,, 
%\end{equation}
%where 
%$$
%\Szero_t := \frac{p^2 t}{2m}\,, \qquad \sigma_{t}:=\sigmazero+\frac{it}{2m\sigmazero}\,,
%$$
%and $e^{itL_{0}}$ is the realization in $L^{\infty}(\RE^2)$ of the strongly continuous (in $L^{2}(\RE^{2})$) group of evolution generated by the self-adjoint operator 
%$$
%L_{0}:=-\,i\,X_{0} \cdot \nabla\,,\qquad X_{0}(q,p):=\left(\frac{p}{m}\,,0\right) ,
%$$ 
%i.e.,  $e^{itL_{0}}f(q,p) = f(q+\frac{p}{m}\,t,p)$.

\begin{proposition}\label{propSemi}
For any $\psih \in L^2(\RE)$ of the form \eqref{CohSt} with $qp\not=0$, there holds
\begin{align}
\big(e^{-i \frac{t}{\hbar} \H_\beta} \psih \big)(x) & = \big(e^{-i \frac{t}{\hbar} \H_0} \psih \big)(x) + \HE(q p)\sgn(x) \, F^{\hbar}_{+,t}\big(\!-\sgn(q) |x|\big) + \HE(-q p)\sgn(x) \, F^{\hbar}_{-,t}\big(\!- \sgn(q) |x|\big) \nonumber \\
& \qquad + E^{\hbar}_{1,t}(x) + E^{\hbar}_{2,t}\big(x\big) + E^{\hbar}_{{\beta},t}(x) \,, \label{psital}
\end{align}
where we set
\begin{align}
F^{\hbar}_{\pm,t}(x) \,:=\, & \,\frac1{\sqrt{2\pi}} \int_{\RE}\! dk\; e^{-i \frac{\hbar t}{2m}\, k^2} e^{i k x}\, R_{\pm}(k)\; \psiht(k)\,, 
\label{F+t}
%\\
%F^{\hbar}_{-,t}(x) \,:=\, & \,\frac1{\sqrt{2\pi}} \int_{\RE}\! dk\; e^{-i \frac{\hbar t}{2m}\, k^2} e^{i k x}\, R_{-}(k)\; \psiht(k)\,, 
%\label{F-t}
\end{align}
\begin{align}
E^{\hbar}_{1,t}(x) 
& := \frac1{{2\pi}} \int_{\RE}\! dk\; e^{-i \frac{\hbar t}{2m}\, k^2} \Big(\!\sgn(x)\, e^{-i|k||x|}\, |R_{+}(k)|^2 -\, e^{ikx}\, R_{-}(k)\Big)\; \times \nonumber \\
& \hspace{4cm} \times \int_{\RE}\! dy\; \big(\sgn(y)\, e^{i|k||y|} - \sgn(q)\, e^{i \sgn(q) |k| y}\big)\, \psih (y)\,, \label{E1}
\end{align}
\begin{equation}\label{E2}
E^{\hbar}_{2,t}(x) := \frac{\sgn(q p x)}{\sqrt{2\pi}} \int_\RE \! dk\; e^{-i \frac{\hbar t}{2m}\, k^2}\, e^{-i\sgn(q)k|x|}\, \theta\big(\!-\sgn(p)k\big) \, \big[ R_{-}(k) - R_{+}(k)\big]\,\psiht(k)\,,
\end{equation}
and
\begin{equation*}%\label{Epp}
E^{\hbar}_{{\beta},t}(x) := e^{-i {t \over \hbar}\, \lambda_\beta}\, (P_{\beta}\psi^{\hbar})(x)\,. 
\end{equation*}
\end{proposition}
\begin{proof} 
Firstly recall the definitions \eqref{eq: eigenfdeltap} for $\vfi^{\pm}_{k}$ and \eqref{eq: Rpm} for $R_{\pm}(k)$. Besides, notice that $\overline{R_{\pm}(k)} = - \,R_{\mp}(k)$.
%% $\overline{R_{\pm}(k)} = - \,R_{\mp}(k) =  R_{\mp}(-k)$.
%%\BLUE{cancella questa e metti il riferimento alla formula}
%% \begin{gather}\label{eq: eigenfdeltap}
%% \vfi^{\pm}_{k}(x) := {e^{i \,k\,x} \over \sqrt{2\pi}} + R_{\pm}(k)\,\sgn(x)\,{e^{\mp i \,|k|\,|x|} \over \sqrt{2\pi}}  \qquad (k \in \RE \backslash \{0\}) \,, \\
%% R_{\pm}(k) := {{i m\,\beta\,k \over \hbar^2} \over 1 \pm  {i m \,\beta\, |k| \over \hbar^2}} = \pm\, {k \over |k| \pm {i \hbar^2 \over m \beta}}\,, \label{eq: Rpm}
%% \end{gather}
Taking as well into account the results of Section \ref{ss:deltaprime} (see, in particular, Eq. \eqref{qevol}), for any $\psih \in L^2(\RE)$ of the form \eqref{CohSt} with $qp\not=0$, we obtain
\begin{equation}\label{psit1}
\begin{aligned} 
\big(e^{-i \frac{t}{\hbar} \H_\beta}\, \psih \big)(x) 
& = \big(e^{-i \frac{t}{\hbar} \H_0}\, \psih \big)(x) 
- \frac1{{2\pi}} \int_{\RE}\! dk\; e^{-i \frac{\hbar t}{2m}\, k^2}\, e^{ikx}\, R_{-}(k) \int_{\RE}\! dy\; \sgn(y)\, e^{i|k||y|}\, \psih (y)  \\ 
& \qquad + \frac{\sgn(x)}{{2\pi}} \int_{\RE}\! dk\; e^{-i \frac{\hbar t}{2m}\, k^2} e^{-i|k||x|}\, R_{+}(k) \int_{\RE}\! dy\; e^{-iky}\, \psih (y) \\ 
& \qquad + \frac{\sgn(x)}{{2\pi}} \int_{\RE}\! dk\; e^{-i \frac{\hbar t}{2m}\, k^2}\,  e^{-i|k||x|}\, |R_{+}(k)|^2 \int_{\RE}\! dy\; \sgn(y)\, e^{i|k||y|}\, \psih (y) \\
& \qquad + e^{-i {t \over \hbar}\, \lambda_\beta}\, (P_{\beta}\psih)(x)\,.
\end{aligned}
\end{equation}
%%Let us anticipate that in the sequel we shall be primarily concerned with states of the form $\psih = e^{-i {t \over \hbar} \H_0}\, \psih_{q_0,p_0}$ for large values of $t$, either positive or negative. On account of Eqs. \eqref{meanq}, \eqref{meanp}, and \eqref{psifree}, we have that these states are also of the form \eqref{CohSt} and such that $\langle \hat{q} \rangle_{\psih} = q_0 + {p_0 t \over m}$ and $\langle \hat{p} \rangle_{\psih} = p_0$.
%%  With this in mind, we shall hereafter restrict the attention to a generic state of the form \eqref{CohSt} with $p> 0$ and (either positive or negative) $q$ large.
Using the elementary identity 
\begin{align*}
& \int_{\RE}\! dy\; \sgn(y)\, e^{i|k||y|}\, \psih (y) \\
& = \sgn(q) \int_{\RE}\! dy\; e^{i \sgn(q) |k| y}\, \psih (y) + \int_{\RE}\! dy\; \big(\sgn(y)\, e^{i|k||y|} - \sgn(q)\, e^{i \sgn(q) |k| y}\big)\, \psih (y)\,,
\end{align*}
Eq. \eqref{psit1} can be reformulated as follows:
\begin{align*} %\label{psit2}
\big(e^{-i \frac{t}{\hbar} \H_\beta}\, \psih \big)(x) 
& = \big(e^{-i \frac{t}{\hbar} \H_0}\, \psih \big)(x) 
- \frac{\sgn(q)}{\sqrt{2 \pi}} \int_{\RE}\! dk\; e^{-i \frac{\hbar t}{2m}\, k^2}\, e^{ikx}\, R_{-}(k)\, \psiht\big(\!-\!\sgn(q) |k|\big) \nonumber \\ 
& \qquad + \frac{\sgn(x)}{\sqrt{2\pi}} \int_{\RE}\! dk\; e^{-i \frac{\hbar t}{2m}\, k^2}\, e^{-i|k||x|}\, R_{+}(k)\; \psiht(k) \nonumber \\ 
& \qquad + \frac{\sgn(q x)}{\sqrt{2\pi}} \int_{\RE}\! dk\; e^{-i \frac{\hbar t}{2m}\, k^2}\, e^{-i|k||x|}\, |R_{+}(k)|^2 \; \psiht\big(\!-\!\sgn(q) |k|\big) 
%% \nonumber \\ & \qquad 
+ E^{\hbar}_{1,t}(x) + E^{\hbar}_{{\beta},t}(x)\,.
\end{align*}
Noting that $R_{\pm}(s\,k) = s\,R_{\pm}(k)$ for $s \in \{\pm 1\}$, by elementary changes of the integration variables we obtain the following identities:
\[\begin{aligned}
&\frac{\sgn(q)}{\sqrt{2\pi}} \int_{\RE}\! dk\; e^{-i \frac{\hbar t}{2m}\, k^2}\, e^{ikx}\, R_{-}(k)\, \psiht\big(\!-\!\sgn(q) |k|\big) \\
& = -\,\frac{\sgn(q x)}{\sqrt{2\pi}} \int_{\RE}\! dk\; e^{-i \frac{\hbar t}{2m}\, k^2} e^{-ik|x|}\, R_{-}(k)\, \psiht\big(\!-\!\sgn(q) |k|\big)\,; 
\end{aligned}
\]
\[\begin{aligned}
& \frac{\sgn(x)}{\sqrt{2\pi}} \int_{\RE}\! dk\; e^{-i \frac{\hbar t}{2m}\, k^2}\, e^{-i|k||x|}\, R_{+}(k)\; \psiht(k)  \\
& = - \,\frac{\sgn(q x)}{\sqrt{2\pi}} \int_{\RE}\! dk\; e^{-i \frac{\hbar t}{2m}\, k^2}\,  e^{-i|k||x|}\, R_{+}(k)\, \psiht\big(\!-\!\sgn(q) k\big)\,; 
\end{aligned}
\]
\[
\begin{aligned}
& \frac{\sgn(q x)}{\sqrt{2\pi}} \int_{\RE}\!\! dk\; e^{-i \frac{\hbar t}{2m}\, k^2}\, e^{-i|k||x|}\, |R_{+}(k)|^2\, \psiht\big(\!-\!\sgn(q) |k|\big) \\ 
& = \frac{2\,\sgn(qx)}{\sqrt{2\pi}} \int_{0}^\infty\!\!\! dk\; e^{-i \frac{\hbar t}{2m}\, k^2}\, e^{-ik|x|}\, |R_{+}(k)|^2\, \psiht\big(\!-\!\sgn(q) k\big)\,. 
\end{aligned}
\]
From the above relations we infer
\begin{align*} %\label{psit3}
\big(e^{-i \frac{t}{\hbar} \H_\beta}\, \psih \big)(x) 
& = \big(e^{-i \frac{t}{\hbar} \H_0}\, \psih \big)(x) 
+ \frac{\sgn(q x)}{\sqrt{2\pi}} \int_{\RE}\! dk\; e^{-i \frac{\hbar t}{2m}\, k^2}\,  e^{-ik|x|}\, R_{-}(k)\, \psiht\big(\!-\!\sgn(q) |k|\big) \nonumber \\ 
& \qquad - \frac{\sgn(q x)}{\sqrt{2\pi}} \int_{\RE}\! dk\; e^{-i \frac{\hbar t}{2m}\, k^2}\, e^{-i|k||x|}\, R_{+}(k)\, \psiht\big(\!-\!\sgn(q) k \big) \nonumber \\ 
& \qquad + \frac{2\,\sgn(q x)}{\sqrt{2\pi}} \int_{0}^\infty\!\!\! dk\; e^{-i \frac{\hbar t}{2m}\, k^2}\, e^{-ik|x|}\, |R_{+}(k)|^2\, \psiht\big(\!-\!\sgn(q) k\big) 
+ E^{\hbar}_{1,t}(x) + E^{\hbar}_{{\beta},t}(x)\,.
\nonumber 
\end{align*}

Hence, taking into account the basic Identity \eqref{basic} we obtain
%% \[ R_{+}(k) - R_{-}(k) = 2\,\sgn(k)\,|R_{+}(k)|^2 \] 
\begin{align*}
\big(e^{-i \frac{t}{\hbar} \H_\beta}\, \psih \big)(x) 
& = \big(e^{-i \frac{t}{\hbar} \H_0}\, \psih \big)(x) 
+ \frac{\sgn(q x)}{\sqrt{2\pi}} \int_{-\infty}^0\! dk\; e^{-i \frac{\hbar t}{2m}\, k^2}\, e^{-ik|x|}\, R_{-}(k)\, \psiht\big(\!-\!\sgn(q) |k|\big) \nonumber \\ 
& \qquad - \frac{\sgn(q x)}{\sqrt{2\pi}} \int_{-\infty}^0\! dk\; e^{-i \frac{\hbar t}{2m}\, k^2}\, e^{-i|k||x|}\, R_{+}(k)\, \psiht\big(\!-\!\sgn(q) k \big) 
+ E^{\hbar}_{1,t}(x) + E^{\hbar}_{{\beta},t}(x) \nonumber \\ 
& = \big(e^{-i \frac{t}{\hbar} \H_0}\, \psih \big)(x) 
- \frac{\sgn(q x)}{\sqrt{2\pi}} \int_0^{\infty}\! dk\; e^{-i \frac{\hbar t}{2m}\, k^2}\, e^{ik|x|}\, R_{-}(k)\, \psiht\big(\!-\!\sgn(q) k\big) \nonumber \\ 
& \qquad + \frac{\sgn(q x)}{\sqrt{2\pi}} \int_0^{\infty}\! dk\; e^{-i \frac{\hbar t}{2m}\, k^2}\, e^{-ik|x|}\, R_{+}(k)\, \psiht\big(\!\sgn(q) k \big) 
+ E^{\hbar}_{1,t}(x) + E^{\hbar}_{{\beta},t}(x)\,. \nonumber 
\end{align*}
Recalling once more that $R_{\pm}(s\,k) = s\,R_{\pm}(k)$ for $s \in \{\pm 1\}$ and using the basic identity
\begin{equation*}%\label{right}
\HE\big(\pm \sgn(q) k\big)  = \HE(\pm q ) \mp \sgn(q)\, \HE\big(\!- k\big) = \HE(\pm q p) \mp \sgn(q p)\,  \HE\big(\!-\sgn(p) k\big)\,,
\end{equation*}
by a few elementary manipulations we obtain
\begin{align*}
\big(e^{-i \frac{t}{\hbar} \H_\beta}\, \psih \big)(x) 
& = \big(e^{-i \frac{t}{\hbar} \H_0}\, \psih \big)(x) 
- \frac{\sgn(q x)}{\sqrt{2\pi}} \int_\RE \! dk\; e^{-i \frac{\hbar t}{2m}\, k^2}\, e^{ik|x|}\, \theta(k)\,R_{-}(k)\, \psiht\big(\!-\!\sgn(q) k\big) \nonumber \\ 
& \qquad + \frac{\sgn(q x)}{\sqrt{2\pi}} \int_\RE \! dk\; e^{-i \frac{\hbar t}{2m}\, k^2}\, e^{-ik|x|}\, \theta(k)\,R_{+}(k)\, \psiht\big(\!\sgn(q) k \big) 
+ E^{\hbar}_{1,t}(x) + E^{\hbar}_{{\beta},t}(x) \nonumber \\
& = \big(e^{-i \frac{t}{\hbar} \H_0}\, \psih \big)(x) 
+ \frac{\sgn(x)}{\sqrt{2\pi}} \int_\RE \! dk\; e^{-i \frac{\hbar t}{2m}\, k^2}\,e^{-i\sgn(q)k|x|}\,\theta\big(-\sgn(q)k\big)\, R_{-}(k)\,\psiht(k) \nonumber \\ 
& \qquad + \frac{\sgn(x)}{\sqrt{2\pi}} \int_\RE \! dk\; e^{-i \frac{\hbar t}{2m}\, k^2}\, e^{-i\sgn(q) k|x|}\,\theta\big(\sgn(q) k\big)\, R_{+}(k)\, \psiht(k) 
+ E^{\hbar}_{1,t}(x) + E^{\hbar}_{{\beta},t}(x) \nonumber \\ 
& = \big(e^{-i \frac{t}{\hbar} \H_0}\, \psih \big)(x) 
+ \frac{\theta(qp)\,\sgn(x)}{\sqrt{2\pi}} \int_\RE \! dk\; e^{-i \frac{\hbar t}{2m}\, k^2}\, e^{-i\sgn(q)k|x|}\, R_{+}(k)\;\psiht(k) \nonumber \\ 
& \qquad + \frac{\theta(-qp)\,\sgn(x)}{\sqrt{2\pi}} \int_\RE \! dk\; e^{-i \frac{\hbar t}{2m}\, k^2}\, e^{-i\sgn(q)k|x|}\, R_{-}(k)\;\psiht(k)\nonumber \\ 
& \qquad -\, \frac{\sgn(q p x)}{\sqrt{2\pi}} \int_\RE \! dk\; e^{-i \frac{\hbar t}{2m}\, k^2}\, e^{-i\sgn(q)k|x|}\, \theta\big(\!-\sgn(p)k\big) \, \big[ R_{+}(k) - R_{-}(k)\big]\,\psiht(k) \nonumber \\ 
& \qquad + E^{\hbar}_{1,t}(x) + E^{\hbar}_{{\beta},t}(x) \,.
\end{align*}
The proof is concluded noting that the latter identity is equivalent to Eq. \eqref{psital}. 
\end{proof}

\begin{lemma}\label{LemmaEpbis}  There exists a constant $C>0$ such that, for any $\psih \in L^2(\RE)$ of the form \eqref{CohSt} with $qp\not=0$,  for all $t\in\RE$ and for all $\eta \in (0,1)$, there holds 
\begin{align}\label{sgnPsit}
\big\|F^{\hbar}_{\pm,t} - R_{\pm}(p/\hbar)\;e^{-i \frac{t}{\hbar} \H_0}\, \psih \big\|_{L^2(\RE)}
\leq C \left[ {\eta \over (1- \eta)}\, \Big({\hbar^3 \over m |\beta p|}\Big) + e^{-\,\eta^2 {p^2 \over 2\hbar |\sigmap|^2}} \right] .
\end{align}
\end{lemma}
\begin{proof} We essentially retrace the same arguments described in \cite[Proof of Lem. 3.3]{AMPA2020}.
Firstly, notice that by unitarity of the Fourier transform we have
\begin{align*}
& \big\|F^{\hbar}_{\pm,t} - R_{\pm}(p/\hbar)\;e^{-i \frac{t}{\hbar} \H_0}\, \psih \big\|_{L^2(\RE)}
%%  = \left\|{1 \over \sqrt{2\pi}}\int_{\RE}\! dk\; e^{-i \frac{\hbar t}{2m}\, k^2} e^{i k x} \big[R_{\pm}(k) - R_{\pm}(p/\hbar)\big]\, \psiht(k) \right\|_{L^2(\RE)}\\ 
%% & = \left\|e^{-i \frac{\hbar t}{2m}\, k^2} \big[R_{\pm}(k) - R_{\pm}(p/\hbar)\big]\, \psiht(k) \right\|_{L^2(\RE,dk)}
= \left(\int_{\RE}\! dk\; \big|R_{\pm}(k) - R_{\pm}(p/\hbar)\big|^2\, \big|\psiht(k)\big|^2\right)^{1/2}\,.
\end{align*}
Recalling the definition of $R_{\pm}(k)$ given in Eq. \eqref{eq: Rpm}, by explicit computations we obtain 
\begin{align*}
\big|R_{\pm}(k) - R_{\pm}(p/\hbar)\big|^2
& = {\Big({m \beta \over \hbar^2}\Big)^{\!2} \left( \big(k - {p/\hbar}\big)^{2} + \Big({m \beta \over \hbar^2}\Big)^{\!2}\,\big(k\,|p|/\hbar - |k|\,p/\hbar \big)^2 \right) \over \left( 1 + \Big({m \beta \over \hbar^2} \Big)^{\!2} \,(p/\hbar)^{2} \right) \left( 1 + \Big({m \beta \over \hbar^2} \Big)^{\!2}\, k^2 \right)} \\
& = {\Big({m \beta p \over \hbar^3}\Big)^{\!2} \left( \big(\hbar k/p - 1\big)^{2} + \big(1 - \sgn(k\,p)\big)^2\,\Big({m \beta p \over \hbar^3}\Big)^{\!2}\, (\hbar\, k/p)^2 \right) \over \left( 1 + \Big({m \beta p \over \hbar^3} \Big)^{\!2} \right) \left( 1 + \Big({m \beta p \over \hbar^3} \Big)^{\!2}\, (\hbar k/p)^2 \right)}\,.
\end{align*}
Starting from here, one obtains the following for any $\eta \in (0,1)$:\footnote{On one hand notice that, for all $b \in \RE $, there holds
%% b := {m \beta p \over \hbar^3}
%% \xi := \hbar k/p
\begin{align*}
& \sup_{\xi \in \RE} {b^2\! \left( (\xi \!-\! 1)^{2} \!+\! (1 \!-\! \sgn\xi)^2\,b^2 \xi^2 \right) \over (1 + b^2)\,(1 + b^2 \xi^2)}
\leq {b^2 \over 1 \!+\! b^2} \, \sup_{\xi \in \RE} \left({(\xi - 1)^{2} \over 1 \!+\! b^2 \xi^2} + {4\,b^2 \xi^2 \over 1 \!+\! b^2 \xi^2}\right)\!
\leq 4 + {b^2 \over 1 \!+\! b^2}\,\max_{\xi \in \RE} {(\xi - 1)^{2} \over 1 \!+\! b^2 \xi^2} 
= 4 + {b^2 (1 \!+\!b^{-2}) \over 1 \!+\! b^2} = 5\,.
\end{align*}
On the other hand, for any given $b \in \RE $, $\eta \in (0,1)$ and for all $|\xi - 1| \leq \eta$ (which ensures $\xi > 0$) we have
\begin{align*}
{b^2 \left( (\xi - 1)^{2} + (1 - \sgn\xi)^2\,b^2 \xi^2 \right) \over (1 + b^2)\,(1 + b^2 \xi^2)} = {(\xi - 1)^{2} \over (1 + 1/b^2)\,(1 + b^2 \xi^2)}  \leq {(\xi - 1)^{2} \over b^2 \xi^2} \leq {(\xi - 1)^{2} \over (1-\eta)^2 b^2}\,.
\end{align*}
}
\begin{gather}
\big|R_{\pm}(k) - R_{\pm}(p/\hbar)\big|^2 \leq 5 \,, \qquad \mbox{for all\, $k \in \RE$}\;; \label{eq: RpmIneq1} \\
\big|R_{\pm}(k) - R_{\pm}(p/\hbar)\big|^2 \leq {1 \over (1- \eta)^2} \Big({\hbar^3 \over m \beta p}\Big)^{\!2}\, (\hbar k/p - 1)^2\,, \qquad \mbox{for\; $|\hbar k /p - 1| \leq \eta$}\; . \label{eq: RpmIneq2}
\end{gather}

So, let us fix $\eta \in (0,1)$ and note that 
\begin{align*}
\big\|F^{\hbar}_{\pm,t} - R_{\pm}(p/\hbar)\;e^{-i \frac{t}{\hbar} \H_0}\, \psih \big\|_{L^2(\RE)}
& \leq \left(\int_{\{\,|\,\hbar k/p  \,-\, 1\,| \,\leq\, \eta \,\}}\hspace{-0.1cm} dk\; \big|R_{\pm}(k) - R_{\pm}(p/\hbar)\big|^2\, \big|\psiht(k)\big|^2\right)^{\!\!1/2} \\
& \qquad + \left(\int_{\{\,|\,\hbar k/p \,-\, 1\,| \,\geq \, \eta \,\}}\hspace{-0.1cm} dk\; \big|R_{\pm}(k) - R_{\pm}(p/\hbar)\big|^2\, \big|\psiht(k)\big|^2\right)^{\!\!1/2} .
\end{align*}
Taking into account the inequalities \eqref{eq: RpmIneq1} and \eqref{eq: RpmIneq2}, we infer respectively
\begin{align*}
& \int_{\{\,|\,\hbar k/p  \,-\, 1\,| \,\geq\, \eta \,\}}\hspace{-0.4cm} dk\; \big|R_{\pm}(k) - R_{\pm}(p/\hbar)\big|^2\, \big|\psiht(k)\big|^2
\leq {5 \over |\sigmap|} \sqrt{{2 \hbar \over \pi}} \int_{\{\,|\,\hbar k/p  \,-\, 1\,| \,\geq\, \eta \,\}}\hspace{-0.5cm} dk\;  e^{-{2 \hbar \over |\sigmap|^2}\,(k- p/\hbar)^2} \\
& \leq {5 \over |\sigmap|} \sqrt{{2 \hbar \over \pi}}\; e^{-\eta^2 {p^2 \over \hbar |\sigmap|^2}} \int_{\RE}\! dk\;  e^{-{\hbar \over |\sigmap|^2}\,(k- p/\hbar)^2}
= 5\sqrt{2}\; e^{-\eta^2 {p^2 \over \hbar |\sigmap|^2}} \,,
\end{align*}
\begin{align*}
& \int_{\{\,|\,\hbar k/p  \,-\, 1\,| \,\leq\, \eta \,\}}\hspace{-0.5cm} dk\; \big|R_{\pm}(k) - R_{\pm}(p/\hbar)\big|^2\, \big|\psiht(k)\big|^2 
\leq {1 \over (1\!-\! \eta)^2} \Big({\hbar^3 \over m \beta p}\Big)^{\!2} \int_{\{\,|\,\hbar k/p  \,-\, 1\,| \,\leq\, \eta \,\}}\hspace{-0.5cm} dk\; (\hbar k/p - 1)^2\, \big|\psiht(k)\big|^2 \\
& \leq {\eta^2 \over (1- \eta)^2}\, \Big({\hbar^3 \over m \beta p}\Big)^{\!2} \int_{\RE}\! dk\; \big|\psiht(k)\big|^2
= {\eta^2 \over (1- \eta)^2}\, \Big({\hbar^3 \over m \beta p}\Big)^{\!2} \;.
\end{align*}
Summing up, the above arguments and the basic relation $\sqrt{a^2 + b^2} \leq a + b$ for $a,b\geq 0$ yield the bound \eqref{sgnPsit}.
\end{proof}

\begin{lemma}\label{LemmaProj} Let $P_{ac}$ and $P_{\beta}$ be defined, respectively, as in Eqs. \eqref{defPac} and \eqref{defPpp}. Then, there exists a constant $C > 0$ such that, for any $\psih \in L^2(\RE)$ of the form \eqref{CohSt} with $qp\not=0$,  there holds
\begin{equation*}
\begin{aligned}
& \big\|P_{\beta} \psih \big\|_{L^2(\RE)} = \big\|P_{ac} \psih - \psih\big\|_{L^2(\RE)} %\label{Pac} 
\\
& \leq C \left({\hbar^{5} |\sigmaq|^2 \over m^2\,\beta^2}\right)^{\!\!1/4} e^{{\hbar^5 |\sigmaq|^2 \over m^2 \beta^2}}
\left({1 \over \sqrt{|\sigmaq\,\sigmap|}}\, e^{- {p^2 \over \hbar |\sigmap|^2} - {\hbar^2 \over m |\beta|} \left(|q| + {2\, \sgn(q p)\, |p| \Im[\sigmap \overline{\sigmaq}] \over |\sigmap|^2}\right) 
- {\hbar^5 |\sigmaq|^2 \over m^2 \beta^2} \left(1 - {1 \over |\sigmaq\, \sigmap|^2} \right)}
+ e^{- \frac{q^2}{4\hbar|\sigmaq|^2}} \right) .
\end{aligned}
\end{equation*}
In particular, for any $\psih = \psih_{\sigmazero,\xi}$ of the form \eqref{psi0} with $qp\not=0$ there holds
\begin{equation}
\begin{aligned}
& \big\|P_{\beta} \psih \big\|_{L^2(\RE)} = \big\|P_{ac} \psih - \psih\big\|_{L^2(\RE)} \leq C \left({\hbar^{5} \sigma_0^2 \over m^2 \beta^2}\right)^{\!\!1/4} e^{{\hbar^5 \sigma_0^2 \over m^2 \beta^2}} \left( e^{- {\sigmazero^2 p^2 \over \hbar}} + e^{- \frac{q^2}{4\hbar \sigmazero^2}} \right) . \label{Pacsi0}
\end{aligned}
\end{equation}
\end{lemma}
\begin{proof}
Recalling that $\vfi_{\beta}$ is a normalized eigenfunction such that $\|\vfi_{\beta}\|_{L^2(\RE)} = 1$, from Eq. \eqref{defPpp} it readily follows
$$
\big\|P_{\beta} \psih \big\|_{L^2(\RE)} =  \left|\int_{\RE}\! dx\; \vfi_{\beta}(x)\; \psih(x)\right| .
$$
By direct computations we obtain
\begin{align*}
& \int_{\RE}\! dx\; \vfi_{\beta}(x)\; \psih(x) 
= {\hbar \over (2\pi \hbar)^{1/4} \sqrt{m\,|\beta|\,\sigmaq}} \int_{\RE}\! dx\;\sgn(x)\;e^{-{\hbar^2 \over m\,|\beta|}\,|x|} \; e^{- \frac{\sigmap}{4\hbar\sigmaq} (x-q)^2 + i \frac{p}{\hbar} (x-q)} \\
& = {\sgn(q)\, \hbar \over (2\pi \hbar)^{1/4} \sqrt{m\,|\beta|\,\sigmaq}} \int_{\RE}\! dy\;\sgn(y)\;e^{-{\hbar^2 \over m\,|\beta|}\,|y|} \; e^{- \frac{\sigmap}{4\hbar\sigmaq} (y-|q|)^2 +  i\, \frac{\sgn(q)\,p}{\hbar} (y-|q|)} \\
& = {\sgn(q)\, \hbar \over (2\pi \hbar)^{1/4} \sqrt{m\,|\beta|\,\sigmaq}} 
\left[\int_{\RE}\! dy\;e^{-{\hbar^2 \over m\,|\beta|}\,y} \; e^{- \frac{\sigmap}{4\hbar\sigmaq} (y-|q|)^2 +  i\, \frac{\sgn(q)\,p}{\hbar} (y-|q|)} \right. \\
&\hspace{4cm} \left. +
\int_{\RE}\! dy \left(\sgn(y)\,e^{-{\hbar^2 \over m\,|\beta|}\,|y|} - e^{-{\hbar^2 \over m\,|\beta|}\,y}\right) e^{- \frac{\sigmap}{4\hbar\sigmaq} (y-|q|)^2 +  i\, \frac{\sgn(q)\,p}{\hbar} (y-|q|)}
\right] .
\end{align*}
On the one hand, keeping in mind our assumptions about the covariance parameters $\sigmaq,\sigmap$ and evaluating explicitly the Gaussian integral we get
\begin{align*}
& \left|{\sgn(q)\, \hbar \over (2\pi \hbar)^{1/4} \sqrt{m\,|\beta|\,\sigmaq}}  \int_{\RE}\! dy\;e^{-{\hbar^2 \over m\,|\beta|}\,y} \; e^{- \frac{\sigmap}{4\hbar\sigmaq} (y-|q|)^2 +  i\, \frac{\sgn(q)\,p}{\hbar} (y-|q|)} \right| \\
%% & = \left|{\sgn(q)\, \hbar \over (2\pi \hbar)^{1/4} \sqrt{m\,|\beta|\,\sigmaq}} \sqrt{2\pi \hbar} \sqrt{{2\sigmaq \over \sigmap}}\;e^{- {p^2 \sigmaq \over \hbar  \sigmap} - {\hbar^2 \over m |\beta|} \left(|q| + \sgn(q)\,i\,{2 p \sigmaq \over \sigmap} \right) + {\hbar^{5} \sigmaq \over m^2 |\beta|^2 \sigmap}}\right| \\
& = \left({8\pi\,\hbar^{5} \over m^2\,\beta^2\,|\sigmap|^2}\right)^{\!\!1/4} e^{- {p^2 \over \hbar |\sigmap|^2} - {\hbar^2 \over m |\beta|} \left(|q| + {2\, \sgn(q p)\, |p| \Im[\sigmap \overline{\sigmaq}] \over |\sigmap|^2}\right) + {\hbar^{5} \over m^2 \beta^2 |\sigmap|^2}}\,.
\end{align*}
On the other hand we have
\begin{align*}
& \left|{\sgn(q)\, \hbar \over (2\pi \hbar)^{1/4} \sqrt{m\,|\beta|\,\sigmaq}} 
\int_{\RE}\! dy \left(\sgn(y)\,e^{-{\hbar^2 \over m\,|\beta|}\,|y|} - e^{-{\hbar^2 \over m\,|\beta|}\,y}\right) e^{- \frac{\sigmap}{4\hbar\sigmaq} (y-|q|)^2 +  i\, \frac{\sgn(q)\,p}{\hbar} (y-|q|)}
\right| \\
& \leq {\hbar \over (2\pi \hbar)^{1/4} \sqrt{m\,|\beta|\,|\sigmaq|}} 
\int_{- \infty}^{0}\! dy \left(e^{{\hbar^2 \over m\,|\beta|}\,y} + e^{-{\hbar^2 \over m\,|\beta|}\,y}\right)
e^{- \frac{1}{4\hbar|\sigmaq|^2} (y-|q|)^2} \\
%% & \leq {\hbar \over (2\pi \hbar)^{1/4} \sqrt{m\,|\beta|\,|\sigmaq|}}\,e^{- \frac{q^2}{4\hbar|\sigmaq|^2}} \int_{0}^{+\infty}\! dy \left(e^{-{\hbar^2 \over m\,|\beta|}\,y} + e^{{\hbar^2 \over m\,|\beta|}\,y}\right) e^{- \frac{1}{4\hbar|\sigmaq|^2}\,y^2} \\
& \leq {2\,\hbar \over (2\pi \hbar)^{1/4} \sqrt{m\,|\beta|\,|\sigmaq|}}\,e^{- \frac{q^2}{4\hbar|\sigmaq|^2}} \int_{\RE}\! dy\; e^{- \frac{1}{4\hbar|\sigmaq|^2}\,y^2 + {\hbar^2 \over m\,|\beta|}\,y}
%% \\ &
= \left({128\,\pi\,\hbar^{5} |\sigmaq|^2 \over m^2\,\beta^2}\right)^{\!\!1/4}\,e^{- \frac{q^2}{4\hbar|\sigmaq|^2} + {\hbar^5 |\sigmaq|^2 \over m^2 \beta^2}}\,.
\end{align*}
The above arguments suffice to infer the thesis.
\end{proof}

\begin{remark} If $p = 0$, recalling the definition of error function and the asymptotic expansions of the latter (see, \emph{e.g.}, \cite[Ch. 7]{NIST}), it can be shown by explicit computations that in the semiclassical limit there holds
\begin{equation*}
 \big\|P_{\beta} \psih \big\|_{L^2(\RE)} = \big\|P_{ac} \psih - \psih\big\|_{L^2(\RE)} = O\left( \!\Big({\hbar^{5} \sigma_0^2 \over m^2\,\beta^2}\Big)^{\!1/4}\right) . %\label{Pacsi0p0}
\end{equation*}
\end{remark}

\begin{lemma}\label{LemmaE1bis} 
 There exists a constant $C>0$ such that, for any $\psih \in L^2(\RE)$ of the form \eqref{CohSt} with $qp\not=0$,  and for all $t \in \RE$, there holds
\begin{equation*}%\label{boundE1}
\big\|E^{\hbar}_{1,t}\big\|_{L^2(\RE)} \leq C\, e^{- {q^2 \over 4\hbar|\sigmaq|^2}}\,.
\end{equation*}
\end{lemma}
\begin{proof} Firstly, let us remark that the definition \eqref{E1} of $E_{1,t}$ can be reformulated as follows, recalling that $R_{\pm}(s\,k) = s\,R_{\pm}(k)$ for $s \in \{\pm 1\}$ and using the basic Identity \eqref{basic}:
\begin{align*}
E^{\hbar}_{1,t}(x) 
%% & = \frac1{{2\pi}} \int_{\RE}\! dk\; e^{-i \frac{\hbar t}{2m}\, k^2} \Big(\!\sgn(x)\, e^{-i|k||x|}\, |R_{+}(k)|^2 -\, e^{ikx}\, R_{-}(k)\Big)\; \times \nonumber \\
%% & \hspace{4cm} \times \int_{\RE}\! dy\; \big(\sgn(y)\, e^{i|k||y|} - \sgn(q)\, e^{i \sgn(q) |k| y}\big)\, \psih (y) \\
& = -\,\frac{\sgn(q x)}{{2\pi}} \int_{\RE}\! dk\; e^{-i \frac{\hbar t}{2m}\, k^2} \Big( e^{-i|k||x|}\, |R_{+}(k)|^2 +  e^{-i k|x|}\, R_{-}(k)\Big)\; \times \nonumber \\
& \hspace{4cm} \times \int_{\RE}\! dy\; \big( \sgn(y)\, e^{i|k||y|} + e^{- i  |k| y}\big)\, \psih \big(\!-\sgn(q)y\big) \\
& = -\,\frac{\sgn(q x)}{{2\pi}} \int_{0}^{\infty}\! dk\; e^{-i \frac{\hbar t}{2m}\, k^2} \Big( e^{-i k |x|}\, 2\,|R_{+}(k)|^2 + (e^{-i k |x|} - e^{i k |x|})\, R_{-}(k)\Big)\; \times \nonumber \\
& \hspace{4cm} \times \int_{0}^{\infty}\! dy\; \big(e^{i k y} + e^{- i k y}\big)\, \psih \big(\!-\sgn(q)y\big) \\
& = -\,\frac{\sgn(q x)}{{\pi}} \int_{0}^{\infty}\! dk\; e^{-i \frac{\hbar t}{2m}\, k^2} \Big( e^{-i k |x|}\, R_{+}(k) - e^{i k |x|}\, R_{-}(k)\Big) \int_{0}^{\infty}\! dy\; \cos(k y)\, \psih \big(\!-\sgn(q)\,y\big)\,.
\end{align*}

To proceed, notice that
\begin{align*}
e^{- i k |x|}\, R_{+}(k) - e^{i k |x|}\,R_{-}(k) 
%% & = {e^{-i k |x|}\,k \over |k| - i\,{\hbar^2 \over m \beta}} + {e^{i k |x|}\,k \over |k| + i\,{\hbar^2 \over m \beta}} \\
& = {k^2\, (e^{i k |x|} + e^{-i k |x|}) \over k^2 + \big({\hbar^2 \over m \beta}\big)^2}
- {i\,\hbar^2 \over m \beta}\, {k\, (e^{i k |x|} - e^{-i k |x|}) \over k^2 + \big({\hbar^2 \over m \beta}\big)^2} \qquad
\mbox{for\, $k > 0$}\,,
\end{align*}
and that the latter expression is an even function of $k$, for $k \in \RE$. Notice also that the integral w.r.t. $y$ gives an even function of $k$ as well. Thus, by symmetry arguments we obtain
\begin{align*}
E^{\hbar}_{1,t}(x) 
& = -\,\frac{\sgn(q x)}{{2\pi}} \int_{\RE}\! dk\; e^{-i \frac{\hbar t}{2m}\, k^2} \left({k^2\, (e^{i k |x|} + e^{-i k |x|}) \over k^2 + \big({\hbar^2 \over m \beta}\big)^2}
- {i\,\hbar^2 \over m \beta}\, {k\, (e^{i k |x|} - e^{-i k |x|}) \over k^2 + \big({\hbar^2 \over m \beta}\big)^2}\right)\; \times \nonumber \\
& \hspace{4cm} \times \int_{0}^{\infty}\! dy\; \cos(k y)\, \psih \big(\!-\sgn(q)\,y\big) \\
& = -\,\frac{\sgn(q x)}{{\pi}} \int_{\RE}\! dk\; e^{-i \frac{\hbar t}{2m}\, k^2}\,e^{i k |x|}\,{k^2 - {i\,\hbar^2 \over m \beta}\,k \over k^2 + \big({\hbar^2 \over m \beta}\big)^2} \int_{0}^{\infty}\! dy\; \cos(k y)\, \psih \big(\!-\sgn(q)\,y\big) \\
& = \frac{\sgn(q x)}{{\pi}} \int_{\RE}\! dk\; e^{-i \frac{\hbar t}{2m}\, k^2}\,e^{i k |x|}\,{k - {i\,\hbar^2 \over m \beta} \over k^2 + \big({\hbar^2 \over m \beta}\big)^2} \int_{0}^{\infty}\! dy\; \sin(k y)\; \partial_y\psih \big(\!-\sgn(q)\,y\big)\,.
\end{align*}
where the last identity is easily derived integrating by parts w.r.t. $y$ and noting that the boundary terms vanish.

Then, by the elementary inequality $\| \psi (|\,\cdot\,|)\|_{L^2(\RE)}^2 \leq 2\,\| \psi \|_{L^2(\RE)}^2$ and by the unitarity of the Fourier transform it follows that 
\begin{align*}
& \|E^{\hbar}_{1,t}\|_{L^2(\RE)}^2 
\leq 2 \left\| \sqrt{\frac{2}{\pi}}\;{1 \over k + {i\,\hbar^2 \over m \beta}} \int_{0}^{\infty}\! dy\; \sin(k y)\; \partial_y \psih \big(\!-\sgn(q)\,y\big) \right\|_{L^2(\RE,dk)}^2 \\ 
& = \frac{2}{\pi} \int_{\RE}\!\!\! dk\; {1 \over k^2 \!+\! \big({\hbar^2 \over m \beta}\big)^{\!2}}
\int_{0}^{\infty}\!\!\!\! dy \int_{0}^{\infty}\!\!\!\! dy' \Big(\!\cos\!\big(k(y\!-\!y')\big) - \cos\!\big(k(y\!+\!y')\big)\!\Big) \partial_{y}\psih \big(\!-\sgn(q)\,y\big) \overline{\partial_{y'}\psih \big(\!-\sgn(q)\,y'\big)}\, .
\end{align*}
From here and from the identity (see \cite[p. 424, Eq. 3.723.2]{GR})
\[
\int_{\RE}\!\! dk\; {\cos(k \xi) \over k^2 \!+\! \big({\hbar^2 \over m \beta}\big)^{\!2}} = {\pi\,m \beta \over \hbar^2}\, e^{-\, {\hbar^2 \over m \beta} |\xi|}\;,
\]
it follows
\begin{align}
\|E^{\hbar}_{1,t}\|_{L^2(\RE)}^2 \leq \mathcal{I}^{\hbar}_{1} +  \mathcal{J}^{\hbar}_{1} \,,\label{E1IJ}
\end{align}
where we put
\begin{align*}
 \mathcal{I}^{\hbar}_{1} &  :=  {2 m \beta \over \hbar^2} \int_{0}^{\infty}\!\! dy\, \int_{0}^{\infty}\!\! dy'\,e^{-\, {\hbar^2 \over m \beta}\, |y - y'|}\, \partial_{y}\psih \big(\!-\sgn(q)y\big)\, \overline{\partial_{y'}\psih \big(\!-\sgn(q)y'\big)}\;,  \\ 
  \mathcal{J}^{\hbar}_{1} & := - \,{2 m \beta \over \hbar^2} \int_{0}^{\infty}\!\! dy\, \int_{0}^{\infty}\!\! dy'\, e^{-\, {\hbar^2 \over m \beta}\, (y + y')}\, \partial_{y}\psih \big(\!-\sgn(q)y\big)\, \overline{\partial_{y'}\psih \big(\!-\sgn(q)y'\big)} \;.
\end{align*}
Via repeated integration by parts and a few elementary manipulations, the latter definitions can be rephrased as follows:
\begin{align*}
 \mathcal{I}^{\hbar}_{1} 
& = {2\, m \beta \over \hbar^2} \int_{0}^{\infty}\!\! dy\; \partial_{y}\psih \big(\!-\sgn(q)y\big) \left[
e^{-\, {\hbar^2 \over m \beta}\, y} \int_{0}^{y}\! dy'\,e^{{\hbar^2 \over m \beta}\, y'}\, \overline{\partial_{y'}\psih \big(\!-\sgn(q)y'\big)} \right. \\
& \left.\hspace{6cm} +\; e^{{\hbar^2 \over m \beta}\, y} \int_{y}^{\infty}\!\! dy'\,e^{-\, {\hbar^2 \over m \beta}\, y'}\, \overline{\partial_{y'}\psih \big(\!-\sgn(q)y'\big)}\,\right] \\
%%%%%%%%%%%%%%%
%% & = {2\, m \beta \over \hbar^2} \int_{0}^{\infty}\!\! dy\; \partial_{y}\psih \big(\!-\sgn(q)y\big) \; \times \\
%% & \qquad \times \left[ e^{-\, {\hbar^2 \over m \beta}\, y} \left(\! \left. e^{{\hbar^2 \over m \beta}\, y'}\, \overline{\psih \big(\!-\sgn(q)y'\big)}\, \right|_{0}^{y} - \int_{0}^{y}\! dy'\,\Big(\partial_{y'} e^{{\hbar^2 \over m \beta}\, y'}\Big)\, \overline{\psih \big(\!-\sgn(q)y'\big)}  \right) \right. \\
%% & \left.\hspace{2cm} +\; e^{{\hbar^2 \over m \beta}\, y} \left(\! \left. e^{-\, {\hbar^2 \over m \beta}\, y'}\, \overline{\psih \big(\!-\sgn(q)y'\big)}\, \right|_{y}^{\infty} - \int_{y}^{\infty}\!\! dy'\,\Big(\partial_{y'} e^{-\, {\hbar^2 \over m \beta}\, y'}\Big)\, \overline{\psih \big(\!-\sgn(q)y'\big)} \right) \right] \\
%%%%%%%%%%%%%%%
& = -\,2 \int_{0}^{\infty}\!\! dy\; \partial_{y}\psih \big(\!-\sgn(q)y\big) \left[
{m \beta \over \hbar^2}\,e^{-\, {\hbar^2 \over m \beta}\, y}\, \overline{\psih(0)}
+ e^{-\, {\hbar^2 \over m \beta}\, y} \int_{0}^{y}\! dy'\,e^{{\hbar^2 \over m \beta}\, y'}\, \overline{\psih \big(\!-\sgn(q)y'\big)}  \right. \\
& \left.\hspace{8.5cm} 
-\; e^{{\hbar^2 \over m \beta}\, y} \int_{y}^{\infty}\!\! dy'\,e^{-\, {\hbar^2 \over m \beta}\, y'}\, \overline{\psih \big(\!-\sgn(q)y'\big)} 
\right] \\
%%%%%%%%%%%%%%%
& = {2 m \beta \over \hbar^2}\,|\psih(0)|^2\!
+ 4 \int_{0}^{\infty}\!\!\! dy\, \big|\psih \big(\!-\sgn(q)y\big)\big|^2\! \\
& \qquad - 2 \int_{0}^{\infty}\!\!\! dy\, e^{-\, {\hbar^2 \over m \beta}\, y} \Big[\overline{\psih(0)}\,\psih \big(\!-\sgn(q)y\big) + \psih(0)\, \overline{\psih \big(\!-\sgn(q)y\big)} \Big] \\
& \qquad -\, {2\hbar^2 \over m \beta} \int_{0}^{\infty}\!\! dy\; \psih \big(\!-\sgn(q)y\big)\, e^{-\, {\hbar^2 \over m \beta}\, y} \int_{0}^{y}\! dy'\,e^{{\hbar^2 \over m \beta}\, y'}\, \overline{\psih \big(\!-\sgn(q)y'\big)}  \\
& \qquad - {2\hbar^2 \over m \beta} \int_{0}^{\infty}\!\! dy\; \psih \big(\!-\sgn(q)y\big)\,  e^{{\hbar^2 \over m \beta}\, y} \int_{y}^{\infty}\!\! dy'\,e^{-\, {\hbar^2 \over m \beta}\, y'}\, \overline{\psih \big(\!-\sgn(q)y'\big)}\;;
\end{align*}
\begin{align*}
\mathcal{J}^{\hbar}_{1} 
& = - \,{2 m \beta \over \hbar^2} \left|\int_{0}^{\infty}\!\! dy\; e^{-\, {\hbar^2 \over m \beta}\, y}\, \partial_{y}\psih \big(\!-\sgn(q)y\big)\right|^{2} \\
& = - \,{2 m \beta \over \hbar^2} \left| -\, \psih(0) + {\hbar^2 \over m \beta} \int_{0}^{\infty}\!\! dy\; e^{-\, {\hbar^2 \over m \beta}\, y}\, \psih \big(\!-\sgn(q)y\big) \right|^{2} \\
& = -\,{2 m \beta \over \hbar^2}\, |\psih(0)|^2
 + \,2 \int_{0}^{\infty}\!\! dy\; e^{-\, {\hbar^2 \over m \beta}\, y}\, \Big(\overline{\psih(0)}\,\psih \big(\!-\sgn(q)y\big) + \psih(0)\,\overline{\psih \big(\!-\sgn(q)y\big)}\Big) \\
& \qquad
-\, {2 \hbar^2 \over m \beta} \int_{0}^{\infty}\!\! dy \int_{0}^{\infty}\!\! dy'\, e^{-\, {\hbar^2 \over m \beta}\, (y+y')}\, \psih \big(\!-\sgn(q)y\big) \overline{\psih \big(\!-\sgn(q)y'\big)}\,.
\end{align*}
Noting that cancellations occur, from the above relations and from Eq.\,\eqref{E1IJ} we infer
\begin{equation*}%\label{E1UVW}
\|E^{\hbar}_{1,t}\|_{L^2(\RE)}^2 \leq \, \mathcal{U}^{\hbar}_{1} + \mathcal{V}^{\hbar}_{1} + \mathcal{W}^{\hbar}_{1}\,,
\end{equation*}
where we put
\begin{align*}
\mathcal{U}^{\hbar}_{1} & := 4 \int_{0}^{\infty}\!\!\! dy\, \big|\psih \big(\!-\sgn(q)y\big)\big|^2\,, \\
\mathcal{V}^{\hbar}_{1} & := -\, {2\hbar^2 \over m \beta} \int_{0}^{\infty}\!\! dy \int_{0}^{\infty}\!\! dy'\,e^{-\, {\hbar^2 \over m \beta}\, |y-y'|}\,\psih \big(\!-\sgn(q)y\big)\, \overline{\psih \big(\!-\sgn(q)y'\big)}\,, \\
\mathcal{W}^{\hbar}_{1} & := -\, {2 \hbar^2 \over m \beta} \int_{0}^{\infty}\!\! dy \int_{0}^{\infty}\!\! dy'\, e^{-\, {\hbar^2 \over m \beta}\, (y+y')}\, \psih \big(\!-\sgn(q)y\big)\, \overline{\psih \big(\!-\sgn(q)y'\big)}\,.
\end{align*}
Now, keeping in mind the basic identity (cf. Eq. \eqref{CohSt} and the related comments)
\begin{gather*}
\big|\psih\big(\!-\!\sgn(q) y\big)\big| = \frac{1}{(2\pi \hbar)^{1/4} \sqrt{|\sigmaq|}}\; e^{- {(y+|q|)^2 \over 4\hbar|\sigmaq|^2}}\,,
\end{gather*}
by arguments similar to those described in the proof of \cite[Lem. 3.5]{AMPA2020}, we infer the following inequalities:
\begin{align*}
& \big|\mathcal{U}^{\hbar}_{1}\big|
= \frac{4}{\sqrt{2\pi \hbar}\, |\sigmaq|} \int_{0}^{\infty}\!\! dy \; e^{- {(y+|q|)^2 \over 2\hbar|\sigmaq|^2}}
\leq \frac{4\; e^{- {|q|^2 \over 2\hbar|\sigmaq|^2}}}{\sqrt{2\pi \hbar}\, |\sigmaq|} \int_{0}^{\infty}\!\! dy \; e^{- {y^2 \over 2\hbar|\sigmaq|^2}}
= 2\,e^{- {|q|^2 \over 2\hbar|\sigmaq|^2}}\,;
\end{align*}
\begin{align*}
\big|\mathcal{V}^{\hbar}_{1}\big|
& \leq {2\hbar^2 \over m \beta} \int_{0}^{\infty}\!\! dy\;\big|\psih \big(\!-\sgn(q)y\big)\big| \left[e^{-\, {\hbar^2 \over m \beta}\, y}\, \int_{0}^{y}\! dy'\,e^{{\hbar^2 \over m \beta}\,y'} \big|\psih \big(\!-\sgn(q)y'\big)\big| \right. \\
& \hspace{6cm}\left. +\; e^{{\hbar^2 \over m \beta}\,y}\, \int_{y}^{\infty}\!\! dy'\,e^{-\, {\hbar^2 \over m \beta}\, y'} \big|\psih \big(\!-\sgn(q)y'\big)\big| \right] \\
& \leq {2\hbar^2 \over m \beta}\;\frac{e^{- {|q|^2 \over 4\hbar|\sigmaq|^2}}}{(2\pi \hbar)^{1/4} \sqrt{|\sigmaq|}} \int_{0}^{\infty}\!\!\! dy\;\big|\psih \big(\!-\sgn(q)y\big)\big| \left[ e^{-\, {\hbar^2 \over m \beta}\, y}  \int_{0}^{y}\!\! dy'\,e^{{\hbar^2 \over m \beta}\,y'}\! + e^{{\hbar^2 \over m \beta}\,y} \int_{y}^{\infty}\!\!\! dy'\,e^{-\, {\hbar^2 \over m \beta}\, y'}\right] \\
& \leq \frac{4\;e^{- {|q|^2 \over 4\hbar|\sigmaq|^2}}}{(2\pi \hbar)^{1/4} \sqrt{|\sigmaq|}} \int_{0}^{\infty}\!\!\! dy\;\big|\psih \big(\!-\sgn(q)y\big)\big|
\leq \frac{4\;e^{- {|q|^2 \over 2\hbar|\sigmaq|^2}}}{\sqrt{2\pi \hbar}\, |\sigmaq|} \int_{0}^{\infty}\!\!\! dy\;e^{- {y^2 \over 4\hbar|\sigmaq|^2}} 
= 2\sqrt{2}\; e^{- {|q|^2 \over 2\hbar|\sigmaq|^2}}\,;
\end{align*}
\begin{align*}
\big|\mathcal{W}^{\hbar}_{1}\big|
& = {2 \hbar^2 \over m \beta} \left|\int_{0}^{\infty}\!\! dy \, e^{-\, {\hbar^2 \over m \beta}\, y}\, \psih \big(\!-\sgn(q)y\big) \right|^2\!
\leq {2 \hbar^2 \over m \beta} \left(\int_{0}^{\infty}\!\! dy \, e^{-\, {2\hbar^2 \over m \beta}\, y}\! \right)\! \left(\int_{0}^{\infty}\!\! dy \, \big| \psih \big(\!-\sgn(q)y\big) \big|^2\!\right) \\
& \leq \frac{e^{- {|q|^2 \over 2\hbar|\sigmaq|^2}}}{\sqrt{2\pi \hbar}\, |\sigmaq|} \int_{0}^{\infty}\!\! dy\; e^{- {y^2 \over 2\hbar|\sigmaq|^2}} 
= {1 \over 2}\;e^{- {|q|^2 \over 2\hbar|\sigmaq|^2}}\;.
\end{align*}
Summing up, the above relations imply the thesis. 
\end{proof}

\begin{lemma}\label{LemmaE2bis}
There exists a constant $C>0$ such that, for any $\psih \in L^2(\RE)$ of the form \eqref{CohSt} with $qp\not=0$ and for all $t\in\RE$, there holds
\begin{equation*}%\label{boundE2bis}
\big\|E^{\hbar}_{2,t}\big\|_{L^2(\RE)}  \leq C\;e^{-{p^2 \over \hbar |\sigmap|^2}}\, .
\end{equation*}
\end{lemma}
\begin{proof}
Recalling the definition of $E^{\hbar}_{2,t}$ (see Eq.\! \eqref{E2}), by the elementary inequality $\| \psi (\pm |\,\cdot\,|)\|_{L^2(\RE)}^2 \!\leq\! 2\,\| \psi \|_{L^2(\RE)}^2$ and by unitarity of the Fourier transform, we have
\begin{align*}
\|E^{\hbar}_{2,t}\|_{L^2(\RE)}^2
& \leq  2 \left\|\frac1{\sqrt{2\pi}} \int_{\RE}\! dk\; e^{-i \frac{\hbar t}{2m}\, k^2}\,e^{i \sgn(q p) k x}\;\HE(k) \Big[R_{-}(k) - R_{+}(k)\Big] \psiht\big(\!- \sgn(p) k\big)\right\|_{L^2(\RE)}^2 \\
& = 2 \int_{0}^{\infty}\!\!dk\; \big|R_{-}(k) - R_{+}(k)\big|^{2}\, \big|\psiht\big(\!- \sgn(p) k\big)\big|^2\,.
\end{align*}
Moreover, from Eqs. \eqref{eq: Rpm} and \eqref{basic} it follows
\begin{align*}
\big|R_{-}(k) - R_{+}(k)\big|^2 =\, 4\,\big|R_{+}(k)\big|^4 =\, {4\, \big({m \beta k \over \hbar^2}\big)^{4} \over \big(1 + \big({m \beta |k| \over \hbar^2}\big)^2 \big)^2} \leq 4\;\sup_{\xi \in \RE}\, {\xi^4 \over (1+\xi^2)^2} = 4\,.
\end{align*}
Thus, taking into account Identity \eqref{psih} for $\psiht$, we infer
\begin{align*}
\big\|E^{\hbar}_{2,t}\big\|_{L^2(\RE)}^2 
& \leq {8 \over |\sigmap|} \sqrt{{2 \hbar \over \pi}} \int_{0}^{\infty}\!\!dk\; e^{-{2\hbar  (k + |p|/\hbar)^2 \over |\sigmap|^2}}
\leq {8 \over |\sigmap|} \sqrt{{2 \hbar \over \pi}}\; e^{-{2\, p^2 \over \hbar |\sigmap|^2}} \int_{0}^{\infty}\!\!dk\; e^{-{2\hbar k^2 \over |\sigmap|^2}}
= 8 \; e^{-{2\, p^2 \over \hbar |\sigmap|^2}} \,,
\end{align*}
which yields the thesis.
\end{proof}

In the next lemma we collect all the results of the previous lemmata. 

\begin{lemma}\label{l:main} There exists a constant $C>0$ such that for any $\psih \in L^2(\RE)$ of the form \eqref{CohSt} with $qp\not=0$, for all $t\in\RE$,  and for all $\eta \in (0,1)$, there holds
\begin{equation}\label{timedependent1}
\begin{aligned}
& \big\|e^{-i \frac{t}{\hbar} \H_\beta} \psih - \upsilonh_t \big\|_{L^2(\RE)}  \\ 
& \leq  C \Bigg[ {\eta \over (1- \eta)}\, \Big({\hbar^3 \over m |\beta p|}\Big) + e^{-\,\eta^2 {p^2 \over 2\hbar |\sigmap|^2}} + e^{- {q^2 \over 4\hbar|\sigmaq|^2}} + e^{-{p^2 \over \hbar|\sigmap|^2}} \\
& \hspace{1cm} + \left({\hbar^{5} |\sigmaq|^2 \over m^2\,\beta^2}\right)^{\!\!1/4} e^{{\hbar^5 |\sigmaq|^2 \over m^2 \beta^2}}
\left({1 \over \sqrt{|\sigmaq\,\sigmap|}}\, e^{- {p^2 \over \hbar |\sigmap|^2} - {\hbar^2 \over m |\beta|} \left(|q| + {2\, \sgn(q p)\, |p| \Im[\sigmap \overline{\sigmaq}] \over |\sigmap|^2}\right) 
- {\hbar^5 |\sigmaq|^2 \over m^2 \beta^2} \left(1 - {1 \over |\sigmaq\, \sigmap|^2} \right)}
+ e^{- \frac{q^2}{4\hbar|\sigmaq|^2}}\right)\!\Bigg]\,,
\end{aligned}
\end{equation}
where 
\begin{equation}
\begin{aligned}\label{phiht2}
 \upsilonh_t (x) & := \big( e^{-i \frac{t}{\hbar} \H_0} \psih \big)(x)  + \HE(qp)\, R_{+}(p/\hbar)\;\sgn(x)\,\big(e^{-i \frac{t}{\hbar} \H_0}\psih\big)\big(\!-\sgn(q)|x|\big)\\ 
& \qquad + \HE(-qp)\, {R_{-}(p/\hbar)}\;\sgn(x)\,\big(e^{-i \frac{t}{\hbar} \H_0}\psih\big)\big(\!-\sgn(q)|x|\big)\,.
\end{aligned}
\end{equation}
\end{lemma}
\begin{proof} The claim \eqref{timedependent1} follows immediately from Eq.\,\eqref{psital}, together with  the expansions of the terms $F^{\hbar}_{\pm,t}$ in Lemma \ref{LemmaEpbis}, and  the bounds on the remainders $E^{\hbar}_{{\beta},t}$, $E^{\hbar}_{1,t}$, $E^{\hbar}_{2,t}$ in Lemmata \ref{LemmaProj}, \ref{LemmaE1bis}, \ref{LemmaE2bis}.
\end{proof}

\begin{lemma}\label{lemma: sgnpsi} For any $\psih \in L^2(\RE)$ of the form \eqref{CohSt} with $qp\not=0$,  there holds
\begin{gather}
\left\|\sgn(\cdot)\,\psih\big(\sigmaq,\sigmap,q,p; \sgn (q)\,|\cdot|\big) - \sgn(q) \Big(\psih(\sigmaq,\sigmap,q,p;\,\cdot\,) - \psih(\sigmaq,\sigmap,-q,-p;\,\cdot\,)\!\Big) \right\|_{L^2(\RE)}\! \leq e^{- \frac{q^2}{4\hbar|\sigmaq|^2}}\,,\label{psi0mod1} \\
\left\|\psih\big(\sigmaq,\sigmap,q,p;- \sgn (q)\,|\cdot|\big) \right\|_{L^2(\RE)} \leq  e^{- \frac{q^2}{4\hbar|\sigmaq|^2}}\,. \label{psi0mod2}
\end{gather}
\end{lemma}
\begin{proof} 
Taking into account that $\psih(\sigmaq,\sigmap,-q,-p;x) = \psih(\sigmaq,\sigmap,q,p;-x) \equiv \psih(-x)$, using the elementary identities $\sgn(q)\,|x| = \sgn(q x)\,x$, $1 = \HE(q x) + \HE(- q x)$ and $\sgn(q x) = \HE(q x) - \HE(-q x)$, by direct computations we get
\begin{align*}
& \left\|\sgn(\cdot)\,\psih\big(\sigmaq,\sigmap,q,p; \sgn (q)\,|\cdot|\big) - \sgn(q) \Big(\psih(\sigmaq,\sigmap,q,p;\,\cdot\,) - \psih(\sigmaq,\sigmap,-q,-p;\,\cdot\,)\!\Big) \right\|_{L^2(\RE)}^2 \\
& = \int_{\RE}\! dx \left|\psih\big(\sigmaq,\sigmap,q,p; \sgn (q x)\,x\big) - \sgn(q x) \Big(\psih(\sigmaq,\sigmap,q,p;x) - \psih(\sigmaq,\sigmap,q,p;-x)\!\Big) \right|^2 \\
& = \int_{\RE}\! dx \left|\HE(q\, x)\, \psih(\sigmaq,\sigmap,q,p;-x) + \HE(- q\, x)\, \psih(\sigmaq,\sigmap,q,p;x) \right|^2\\
& = 2 \int_{\RE}\! dx\; \HE(q\, x)\, \big|\psih(\sigmaq,\sigmap,q,p;-x)\big|^2 
	= 2 \int_{0}^{+\infty}\!\! dy\; \Big|\psih\big(\sigmaq,\sigmap,q,p;-\sgn(q)\,y\big)\Big|^2 \,.
\end{align*}
From here, noting the identity $\Re(\sigmap/\sigmaq) = |\sigmaq|^{-2}$ (see Eq. \eqref{sxsp1}) and using the inequality $e^{-(a+b)^2} \leq e^{-a^2-b^2}$ for $a,b\geq 0 $, we infer
\begin{align*}
& \left\|\sgn(\cdot)\,\psih\big(\sigmaq,\sigmap,q,p; \sgn (q)\,|\cdot|\big) - \sgn(q) \Big(\psih(\sigmaq,\sigmap,q,p;\,\cdot\,) - \psih(\sigmaq,\sigmap,-q,-p;\,\cdot\,)\!\Big) \right\|_{L^2(\RE)}^2 \\
& = \frac{2}{\sqrt{2\pi \hbar}\; |\sigmaq|} \int_{0}^{+\infty}\!\!\! dy\; e^{- \frac{(y + |q|)^2}{2\hbar|\sigmaq|^2}}
	\leq \frac{2\, e^{- \frac{q^2}{2\hbar|\sigmaq|^2}}}{\sqrt{2\pi \hbar}\; |\sigmaq|} \int_{0}^{+\infty}\!\!\! dy\;  e^{- \frac{y^2}{2\hbar|\sigmaq|^2}} = e^{- \frac{q^2}{2\hbar|\sigmaq|^2}}\,,
\end{align*}
which proves Eq. \eqref{psi0mod1}. Eq. \eqref{psi0mod2} can be derived by similar arguments (cf. \cite[Lem. 3.8]{AMPA2020}).
\end{proof}

\subsection{Proof of Theorem \ref{t:1}\label{ss:3.1}}

At first, in the following proposition we give an explicit formula for the semiclassical limit  evolution of a coherent state. 
\begin{proposition} \label{p:blacksun}
Let ${B}(p) := - \,(2  \beta/\hbar^3)\,p^2$. Then, under the assumptions of Theorem \ref{t:1} there holds
\begin{align*}%\label{blacksun}
& e^{\frac{i}{\hbar}\Szero_{t}} \big(e^{itL_{{B}}}\phi^{\hbar}_{\sigma_{t},x}\big)(\xi)
= \big(e^{-i\frac{t}{\hbar}\H_{0}}\,\psi^{\hbar}_{\sigmazero,\xi}\big)(x) \\
&\qquad -\, \sgn(q)\,\HE(- qp)\,\HE \!\left(t+\frac{m q}{p}\right)R_{-}(p/\hbar)\, \Big( \big(e^{-i\frac{t}{\hbar}\H_{0}}\,\psi^{\hbar}_{\sigmazero,\xi}\big)(x) - \big(e^{-i\frac{t}{\hbar}\H_{0}}\,\psi^{\hbar}_{\sigmazero,\xi}\big)(-x) \Big) \nonumber \\
&\qquad -\, \sgn(q)\,\HE(qp)\,\HE \!\left(-t-\frac{mq}{p}\right) R_{+}(p/\hbar)\, \Big( \big(e^{-i\frac{t}{\hbar}\H_{0}}\,\psi^{\hbar}_{\sigmazero,\xi}\big)(x) - \big(e^{-i\frac{t}{\hbar}\H_{0}}\,\psi^{\hbar}_{\sigmazero,\xi}\big)(-x) \Big) \,. \nonumber 
\end{align*}
\end{proposition}
\begin{proof}
Recall that $\xi = (q,p)$. We start by noticing that, by Eq.\,\eqref{eq: timeevcl}, 
\begin{equation*}
\big(e^{it L_{{B}}} \phi^{\hbar}_{\sigma_{t},x}\big)(\xi)=
\big(e^{i t L_{0}} \phi^{\hbar}_{\sigma_{t},x}\big)(q,p) \,-\, \frac{\HE(- t\,q\,p)\; \HE\big({|p\,t| \over m} - |q|\big)}{1 - \sgn(t)\,{2i\,|p| \over m\,{B}(p)}}\, \Big(\big(e^{itL_{0}} \phi_{\sigma_{t},x}\big)(q,p) \,-\, \big(e^{itL_{0}}\phi_{\sigma_{t},x}\big)(-q,-p)\Big)\,,
\end{equation*}
hence, on account of Identity \eqref{free2}, we infer 
\begin{align*}
e^{\frac{i}{\hbar}\Szero_{t}} \big(e^{itL_{{B}}}\phi^{\hbar}_{\sigma_{t},x}\big)(\xi)
& = \big(e^{-i\frac{t}{\hbar}\H_{0}}\,\psi^{\hbar}_{\sigmazero,\xi}\big)(x) \\
& \qquad -\, \frac{\HE(- t\,q\,p)\; \HE\big({|p\,t| \over m} - |q|\big)}{1 - \sgn(t)\,{2i\,|p| \over m\,{B}(p)}}\, \Big( \big(e^{-i\frac{t}{\hbar}\H_{0}}\,\psi^{\hbar}_{\sigmazero,\xi}\big)(x) - \big(e^{-i\frac{t}{\hbar}\H_{0}}\,\psi^{\hbar}_{\sigmazero,-\xi}\big)(x) \Big)\,. 
\end{align*}

We note that 
\[\begin{aligned}
\big(e^{-i\frac{t}{\hbar}\H_{0}}\,\psi^{\hbar}_{\sigmazero,-\xi}\big)(x) 
=\, &\,  e^{\frac{i}{\hbar}\Szero_{t}} \psi^{\hbar}\Big(\sigmazero+\frac{it}{2m\sigmazero},\sigmazero^{-1},- q-\frac{pt}{m},-p;x\Big)  \\ 
=\, & \,e^{\frac{i}{\hbar}\Szero_{t}}\psi^{\hbar} \Big(\sigmazero+\frac{it}{2m\sigmazero},\sigmazero^{-1},q+\frac{pt}{m},p;-x\Big)  = \big(e^{-i\frac{t}{\hbar}\H_{0}}\,\psi^{\hbar}_{\sigmazero,\xi}\big)(-x) \,,
\end{aligned}
\]
whence, 
\begin{align*}
e^{\frac{i}{\hbar}\Szero_{t}} \big(e^{itL_{{B}}}\phi^{\hbar}_{\sigma_{t},x}\big)(\xi)
& = \big(e^{-i\frac{t}{\hbar}\H_{0}}\,\psi^{\hbar}_{\sigmazero,\xi}\big)(x) \\
& \qquad -\, \frac{\HE(- t\,q\,p)\; \HE\big({|p\,t| \over m} - |q|\big)}{1 - \sgn(t)\,{2i\,|p| \over m\,{B}(p)}}\, \Big( \big(e^{-i\frac{t}{\hbar}\H_{0}}\,\psi^{\hbar}_{\sigmazero,\xi}\big)(x) - \big(e^{-i\frac{t}{\hbar}\H_{0}}\,\psi^{\hbar}_{\sigmazero,\xi}\big)(-x) \Big)\,. 
\end{align*}

To conclude the proof we observe that 
\begin{align*}
\frac{\HE (-tqp)\,\HE \!\left(\frac{|pt|}{m}-|q|\right)}{1 - \sgn(t)\,{2i\,|p| \over m\,{B}(p)}} 
=\, &\, \frac{\HE (t)\,\HE (-qp)\,\HE \!\left(\frac{|p|t}{m}-|q|\right)}{1 - {2i\,|p| \over m\,{B}(p)}} + \frac{\HE(-t)\,\HE(qp)\,\HE \!\left(-\frac{|p|t}{m}-|q|\right)}{1 +\,{2i\,|p| \over m\,{B}(p)}} \nonumber \\
=\, &\, \frac{\HE (-qp)\,\HE \!\left(t-\frac{m |q|}{|p|}\right)}{1-{2i\,|p| \over m\,{B}(p)}} + \frac{\HE (qp)\,\HE \!\left(-t-\frac{m|q|}{|p|}\right)}{1+ {2i\,|p| \over m\,{B}(p)}}
 \nonumber \\ 
=\, &\, \frac{\HE (-qp)\,\HE \!\left(t+\frac{m q}{p}\right)}{1-{2i\,|p| \over m\,{B}(p)}} + \frac{\HE (qp)\,\HE \!\left(-t-\frac{mq}{p}\right)}{1+{2i\,|p| \over m\,{B}(p)}}\;. 
\end{align*}
Notably, setting ${B}(p) := -\,(2  \beta/\hbar^3)\,p^2$ and recalling the definition \eqref{eq: Rpm} of $R_{\pm}(k)$ we obtain
$$
{\HE(\pm qp) \over 1 \pm {2i|p| \over m \,{B}(p)}}
= \HE(\pm qp)\,\Big(\!\pm\,\sgn(p)\,R_{\pm}(p/\hbar)\Big)
= \sgn(q)\,\HE(\pm qp)\,R_{\pm}(p/\hbar)\,.
$$
Summing up, the arguments described above imply the thesis.
\end{proof}

We are now ready to prove Theorem \ref{t:1}. 
\begin{proof}[Proof of Theorem \ref{t:1}] 
First, we use  Lemma \ref{l:main} to approximate the state $e^{-i\frac{t}{\hbar}\H_{\beta}}\,\psi^{\hbar}_{\sigmazero,\xi}$ with $\upsilonh_{\sigmazero,\xi,t}$ defined according to Eq. \eqref{phiht2} by
\begin{equation*}
\begin{aligned}
\upsilonh_{\sigmazero,\xi,t} (x) & := \big( e^{-i \frac{t}{\hbar} \H_0} \psih_{\sigmazero,\xi} \big)(x)  + \HE(qp)\, R_{+}(p/\hbar)\;\sgn(x)\,\big(e^{-i \frac{t}{\hbar} \H_0}\psih_{\sigmazero,\xi}\big)\big(\!-\sgn(q)\,|x|\big)\\ 
& \qquad + \HE(-qp)\, {R_{-}(p/\hbar)}\;\sgn(x)\,\big(e^{-i \frac{t}{\hbar} \H_0}\psih_{\sigmazero,\xi}\big)\big(\!-\sgn(q)\,|x|\big)\,.
\end{aligned}
\end{equation*}
Next, we compare $\upsilonh_{\sigmazero,\xi,t}$ with the expression for $e^{\frac{i}{\hbar}\Szero_{t}}\,e^{itL_{{B}}} \phi^{\hbar}_{\sigma_{t},x}(\xi)$ from Proposition \ref{p:blacksun}. Retracing the arguments described in \cite[Proof of Thm. 1.1)]{AMPA2020}, we infer
\vspace{-0.1cm}
\[\begin{aligned}
\upsilonh_{\sigmazero,\xi,t}(x) & = \big( e^{-i \frac{t}{\hbar} \H_0} \psih_{\sigmazero,\xi} \big)(x) 
+ \HE\Big(t + {m q \over p}\Big)\, \HE(qp)\, R_{+}(p/\hbar)\;\sgn(x)\,\big(e^{-i \frac{t}{\hbar} \H_0}\psih_{\sigmazero,\xi}\big)\big(\!-\sgn(q_t)\,|x|\big)\\ 
& \qquad + \HE\Big(t + {m q \over p}\Big)\, \HE(-qp)\, {R_{-}(p/\hbar)}\;\sgn(x)\,\big(e^{-i \frac{t}{\hbar} \H_0}\psih_{\sigmazero,\xi}\big)\big(\sgn(q_t)\,|x|\big) \\
& \qquad + \HE\Big(\!-t - {m q \over p}\Big)\, \HE(qp)\, R_{+}(p/\hbar)\;\sgn(x)\,\big(e^{-i \frac{t}{\hbar} \H_0}\psih_{\sigmazero,\xi}\big)\big(\sgn(q_t)\,|x|\big)\\ 
& \qquad + \HE\Big(\!-t - {m q \over p}\Big)\, \HE(-qp)\, {R_{-}(p/\hbar)}\;\sgn(x)\,\big(e^{-i \frac{t}{\hbar} \H_0}\psih_{\sigmazero,\xi}\big)\big(\!-\sgn(q_t)\,|x|\big)\,.
\end{aligned}\]

Recall that 
\[
\big(e^{-i\frac{t}{\hbar}\H_{0}}\,\psi^{\hbar}_{\sigmazero,\xi}\big)(x) =   e^{\frac{i}{\hbar}\Szero_{t}}\, \psi^{\hbar}\big(\sigma_t,\sigmazero^{-1},q_t,p;x\big)\,,
\]
with $\sigma_t = \sigmazero+\frac{it}{2m\sigmazero} $ and $q_t =  q + \frac{pt}{m}$. Hence, by Lemma \ref{lemma: sgnpsi}, we deduce:
\begin{gather*}
\left\|\sgn(\cdot)\,\big(e^{-i \frac{t}{\hbar} \H_0}\psih_{\sigmazero,\xi}\big)\big(\!\sgn(q_t)|\cdot|\big)- \sgn(q_t)\Big(
\big(e^{-i\frac{t}{\hbar}\H_{0}}\,\psi^{\hbar}_{\sigmazero,\xi}\big)(x) - \big(e^{-i\frac{t}{\hbar}\H_{0}}\,\psi^{\hbar}_{\sigmazero,\xi}\big)(-x)\Big)\right\|_{L^2(\RE)} \leq e^{- \frac{q_t^2}{4\hbar|\sigma_t|^2}} \,; \\
\left\|\big(e^{-i \frac{t}{\hbar} \H_0}\psih_{\sigmazero,\xi}\big)\big(\!-\sgn(q_t)|\cdot|\big)\right\|_{L^2(\RE)} \leq  e^{- \frac{q_t^2}{4\hbar|\sigma_t|^2}}\,. 
\end{gather*}

Noting that
\begin{align*}
\HE\Big(\!\pm t \pm {m q \over p}\Big)\, \HE(\mp qp)\, \sgn(q_t) = \pm\, \HE\Big(\!\pm t \pm {m q \over p}\Big)\, \HE(qp)\, \sgn(p) = -\,\HE\Big(\!\pm t \pm {m q \over p}\Big)\, \HE(\mp qp)\, \sgn(q)\,, \\
\end{align*}
the previous bounds, together with Proposition \ref{p:blacksun}, imply 
\[
\left\| \upsilonh_{\sigmazero,\xi,t} (x)  - e^{\frac{i}{\hbar}\Szero_{t}} \big(e^{itL_{{B}}}\phi^{\hbar}_{\sigma_{t},(\cdot)}\big)(\xi)\right\|_{L^2(\RE)}  \leq 4\, e^{- \frac{q_t^2}{4\hbar|\sigma_t|^2}}\,. 
\]
Besides the exponential  $e^{- \frac{q_t^2}{4\hbar|\sigma_t|^2}}$, the remaining  terms on the r.h.s. of Eq. \eqref{t1} are a consequence of the fact that we approximated $e^{-i\frac{t}{\hbar}\H_{\beta}}\,\psi^{\hbar}_{\sigmazero,\xi}$ with $\upsilonh_{\sigmazero,\xi,t}$ and of Lemma \ref{l:main} (here employed with $\sigmaq = \sigma_0$, $\sigmap = 1/\sigma_0$, and using as well the basic inequality $e^{- {\hbar^2 |q| \over m |\beta|}} \leq 1$).
\end{proof}

\subsection{Proof of Corollary \ref{c:1}.}\label{proof-coroll}
\begin{proof}%{Proof of Corollary \ref{c:1}.}
Fix $\eta = \underline{h}^{1/2-\lambda}$ in Theorem \ref{t:1}. Then, for $0 < \lambda < 1/2$ and fixed $(q,p) \in \RE^2$ with $q p \neq 0$, the time-independent part on the r.h.s. of inequality \eqref{t1} is bounded by\footnote{Especially, let us mention that we used the upper bound descending from the following chain of inequalities:
\[
\left({\hbar^{5} \sigma_0^2 \over m^2 \beta^2}\right)^{\!\!1/4} e^{{\hbar^5 \sigma_0^2 \over m^2 \beta^2}} \Bigg(e^{- {\sigma_0^2 p^2 \over \hbar}} \!+ e^{- \frac{q^2}{4\hbar \sigma_0^2}}\Bigg)
= \left({\hbar \over (m |\beta p|)^{1/3}}\right)^{\!\!3/2}\! \left({\sigma_0^2 p^2 \over \hbar}\right)^{\!\!1/4} e^{\big({\hbar \over (m |\beta p|)^{1/3}}\big)^{6} {\sigma_0^2 p^2 \over \hbar} } \Bigg(e^{- {\sigma_0^2 p^2 \over \hbar}} + e^{- \frac{q^2}{4\hbar \sigma_0^2}}\Bigg) \\
%% & = \left({\hbar \over (m |\beta p|)^{1/3}}\!\right)^{\!\!3/2}\!\! \left(\!{\sigma_0^2 p^2 \over \hbar}\right)^{\!\!1/4}\! \left[ e^{- {\sigma_0^2 p^2 \over \hbar} \left(1 - \left({\hbar \over (m |\beta p|)^{1/3}}\!\right)^6\right)} + e^{- {\sigma_0^2 p^2 \over \hbar} \left(\frac{q^2}{4 \sigma_0^4 p^2} - \left({\hbar \over (m |\beta p|)^{1/3}}\!\right)^6 \right)} \right] \vspace{-0.2cm}
\]
\[
\leq \underline{h}^{3/2}\! \left(\!{\sigma_0^2 p^2 \over \hbar}\right)^{\!\!1/4} \Bigg[ e^{- (1 - \underline{h}^6)\,{\sigma_0^2 p^2 \over \hbar}} + e^{- \frac{q^2}{4 \sigma_0^4 p^2} \big(1 - \frac{4 \sigma_0^4 p^2}{q^2}\,\underline{h}^6 \big) {\sigma_0^2 p^2 \over \hbar}} \Bigg]
%% & \leq \underline{h}^{3/2}\! \left[ \left(\!{\sigma_0^2 p^2 \over \hbar}\right)^{\!\!1/4} e^{- {\sigma_0^2 p^2 \over 4 \hbar}}  e^{- {\sigma_0^2 p^2 \over 4 \hbar}} + \left(\!{\sigma_0^2 p^2 \over \hbar}\right)^{\!\!1/4} e^{- \frac{q^2}{16 \sigma_0^4 p^2}\, {\sigma_0^2 p^2 \over \hbar}} e^{- \frac{q^2}{16 \sigma_0^4 p^2}\, {\sigma_0^2 p^2 \over \hbar}} \right] 
\leq {16 \over e}\,\max\Big\{1, \frac{4\sigma_0^4 p^2}{q^2}\Big\}\;\underline{h}^{3/2}\! \left(e^{- \frac{1}{4 \underline{h}}} + e^{- \frac{1}{16 \underline{h}}} \right),
\]
where the last inequality follows noting that $\min \big\{1 - \underline{h}^6, 1 -\frac{4 \sigma_0^4 p^2}{q^2}\,\underline{h}^6 \big\} \geq 1/2$ for $\underline{h}$ small enough, and that ${\sup_{\{\xi > 0\}}}\,\big( \xi^{-1/4}\,e^{-\mu/\xi}\big) \leq \xi^{-1/4}\,e^{-\mu/\xi}\big|_{\xi = \mu} = (4/e)\, \mu^{-1}$ for any $\mu > 0$.}
\begin{equation*}
\begin{aligned}
& C \Bigg[ {\underline{h}^{7/2-\lambda} \over 1- \underline{h}^{1/2-\lambda}} + e^{-\,{1 \over 2\underline{h}^{2\lambda}}} + e^{- {1 \over 4 \underline{h}}} + e^{-{1 \over \underline{h}}} + \underline{h}^{3/2}\! \left(e^{- \frac{1}{4 \underline{h}}} + e^{- \frac{1}{16 \underline{h}}} \right)\!\Bigg]
\leq C_{*}\, \underline{h}^{7/2-\lambda} \,,
\end{aligned}
\end{equation*}
for some $C_*>0$ and for all $\underline h < h_*$ with $h_*$ small enough.
%%\vspace{-0.3cm}
%%\vfill\eject\noindent

On the other hand, to take into account the time-dependent term on the r.h.s. of inequality \eqref{t1} it is enough to show that if $|t-t_{coll}(\xi)| \ge c_0\, |t_{coll}(\xi)|\,\sqrt{(7/2-\lambda)\,\underline h\,|\ln \underline h|}$ (for some $c_0 > 0$), then $ q_t^2/(4\hbar|\sigma_t|^2) \ge (7/2-\lambda)\,|\ln \underline h|$. Setting $y = 1 - t/t_{coll}(\xi)$, $a = {4\hbar\sigmazero^2 \over q^2}\,(7/2-\lambda)\,|\ln \underline h|$ and $b = {\hbar \over \sigma_0^2 p^2}\,(7/2-\lambda)\,|\ln \underline h|$, the latter relation can be rephrased as $y^2/(a + b\, (1-y)^2) \ge 1$; a simple calculation shows that this inequality is fulfilled if
\begin{equation}\label{y}
a,b \in (0,1) \qquad \mbox{and} \qquad |y| \ge {b + \sqrt{a + b - a b} \over 1 - b}\,.
\end{equation}
Taking into account that $a + b - ab \le a + b$, $a \leq 4 (7/2-\lambda)\, \underline h\,|\ln \underline h|$ and $b \leq (7/2-\lambda)\, \underline h\,|\ln \underline h|$, it is easy to convince oneself that when $\underline h$ is small enough Eq. \eqref{y} holds true as soon as $|y| > c_0\,\sqrt{(7/2-\lambda)\,\underline h\,|\ln \underline h|}$ for some $c_0 > \sqrt{5}$, which proves Eq. \eqref{ttcoll}.
\end{proof}

\section{Convergence of the wave and scattering operators}

\begin{lemma}\label{lemma: Ompm}
For any $\psih \in L^2(\RE)$ of the form \eqref{CohSt} with $qp\not=0$, there holds 
\begin{align}
\label{Omega_beta}
\big(\Omega_\beta^\pm\, \psih\big) (x) =\, &  \,\psih (x) + \sgn(x) \left[\HE(q p)\, F^{\hbar}_{\pm,0}\big(\!\mp\sgn(q)\,|x|\big)\! + \HE(-q p)\, F^{\hbar}_{\pm,0}\big(\!\pm\sgn(q)\,|x|\big)\right] +  E^{\hbar}_{3,\pm}(x)\,, 
\end{align}
where $F^{\hbar}_{\pm,0} \equiv F^{\hbar}_{\pm,t = 0}$, see Eq. \eqref{F+t}, and
\begin{gather*}
%% F^{\hbar}_{\pm,0}(x) \equiv F^{\hbar}_{\pm,t = 0}(x) \equiv \frac{1}{\sqrt{2\pi}} \int_{\RE}\! dk\; e^{i k x}\, R_{\pm}(k)\; \psiht(k)\,, \\
E^{\hbar}_{3,\pm}(x) := \pm\, {\sgn(q\,x) \over \sqrt{2\pi}} \int_{\RE}\! dk \left( e^{i \sgn(q p)\,k\,|x|} - \,e^{- i \sgn(q p)\,k\,|x|} \right) \HE(k)\,R_{\pm}(k)\, \psiht\big(-\sgn(p)\,k\big)\,.
\end{gather*}
\end{lemma}
\begin{proof} 
First notice that, from \eqref{st-rep}, \eqref{eq: eigenfdeltap} and \eqref{eq: FFpm} it follows
\[
\big(\Omega_\beta^{\pm}\, \psih\big) (x) 
 = \psih(x) + {\sgn(x) \over \sqrt{2\pi}} \int_{\RE}\! dk\;e^{\mp i \,|k|\,|x|}\, R_{\pm}(k)\, \psiht(k)\,.
\]
From here and from the identities $\HE\big(\!\pm \sgn(q) k\big) + \HE\big(\!\mp \sgn(q) k\big) = 1$ and $ \HE\big(\!\pm \sgn(q) k\big)  = \HE(\pm q p) \mp \sgn(q p)\, \HE\big(\!-\sgn(p) k\big)$, we obtain
\begin{align*}
& \big(\Omega_\beta^{\pm}\, \psih\big) (x) 
	= \psih(x) + {\sgn(x) \over \sqrt{2\pi}} \int_{\RE}\! dk\;e^{\mp i \,|k|\,|x|}\Big[\HE\big(\!\pm \sgn(q) k\big) + \HE\big(\! \mp \sgn(q) k\big) \Big]\, R_{\pm}(k)\, \psiht(k) \\
& = \psih(x) + {\sgn(x) \over \sqrt{2\pi}} \int_{\RE}\! dk\,\Big[\HE\big(\!\pm \sgn(q) k\big)\,e^{- i \sgn(q) k |x|} + \HE\big(\!\mp \sgn(q) k\big)\,e^{i \sgn(q) k |x|} \Big]\, R_{\pm}(k)\, \psiht(k) \\
& = \psih(x) 
+ {\sgn(x) \over \sqrt{2\pi}} \int_{\RE}\! dk \left[ \HE(\pm q p)\,e^{-i \sgn(q) k\,|x|} + \HE(\mp q p)\, e^{i \sgn(q) k\,|x|} \right] R_{\pm}(k)\,\psiht(k) \\
& \qquad \mp {\sgn(x)\,\sgn(q p) \over \sqrt{2\pi}} \int_{\RE}\! dk \left( e^{-i \sgn(q) k\,|x|} - \, e^{i \sgn(q) k\,|x|} \right) \HE\big(\!-\sgn(p) k\big)\,R_{\pm}(k)\, \psiht(k) \,,
\end{align*}
which, noting that $R_{\pm}\big(\!-\sgn(p) k\big) = -\sgn(p)\,R_{\pm}(k)$ (see Eq. \eqref{eq: Rpm}), is equivalent to Eq. \eqref{Omega_beta}. 
\end{proof}

\begin{lemma}\label{LemmaE3} 
There exists a constant $C>0$ such that, for any $\psih \in L^2(\RE)$ of the form \eqref{CohSt} with $qp\not=0$, there holds
\begin{equation*}%\label{boundE3}
\big\|E^{\hbar}_{3,\pm}\big\|_{L^2(\RE)} \leq C\, e^{- {p^2 \over \hbar |\sigmap|^{2}}}\,.
\end{equation*}
\end{lemma}
\begin{proof}
By the elementary inequality $\| \psi ( |\cdot|)\|_{L^2(\RE)}^2+\| \psi (-|\cdot|)\|_{L^2(\RE)}^2  \leq  4\,\| \psi \|_{L^2(\RE)}^2$, by unitarity of the Fourier transform, and by the basic bound $| R_{\pm}(k)|\leq 1$, we infer that 
\begin{align*}
\big\|E^{\hbar}_{3,\pm}\big\|_{L^2(\RE)}^2 
& \leq 4  \left\|\frac{1}{\sqrt{2\pi}} \int_\RE\! dk\;  e^{-ik x}\, \HE(k)\,R_{\pm}(k)\, \psiht\big(\!-\sgn(p)k\big) \right\|_{L^2(\RE)}^2 \\
& = 4  \int_{0}^{\infty}\!\!\!dk\; \big| R_{\pm}(k)\big|^{2}\, \big|\psiht\big(\!-\sgn(p) k\big)\big|^2 \leq 4 \int_{0}^{\infty}\!\!\!dk\; \big|\psiht\big(\!-\sgn(p) k\big)\big|^2 \\ 
& = {4 \over |\sigmap|} \left({2 \hbar \over \pi}\right)^{\!\!1/2} \int_{0}^{\infty}\!\!\!dk\; e^{-2 {\hbar (k + |p|/\hbar)^2 \over |\sigmap|^{2}}} 
\leq {4 \over |\sigmap|} \left({2 \hbar \over \pi}\right)^{\!\!1/2} e^{- {2 p^2 \over \hbar |\sigmap|^{2}}} \int_{0}^{\infty}\!\!\!dk\; e^{- {2 \hbar k^2 \over |\sigmap|^{2}}}
= 2\, e^{- {2 p^2 \over \hbar |\sigmap|^{2}}}\,,
\end{align*}
which yields the thesis.
\end{proof}

\subsection{Proof of Theorem \ref{thm: Ompm}}\label{convWops}
%\begin{proof}[Proof of Theorem \ref{thm: Ompm}]xxxx
We first prove claim \eqref{t2_1}.

Preliminarily we apply the classical wave operators $W_{{B}}^{\pm}$, with ${B}(p) := -\,(2 \beta/\hbar^3)\, p^2$ to the state $\phi^{\hbar}_{\sigma_0,\,(\cdot)}(\xi)$, with $\xi = (q,p)$. Recalling the definition \eqref{eq: Rpm} of $R_{\pm}(k)$ and noting the basic identity $\mp \sgn(p) \HE(\mp q p) = \sgn(q)\, \HE(\mp q p)$, from Eq. \eqref{eq: waveopcl} we infer:
\begin{align*}
\big(W_{{B}}^{\pm} \phi^{\hbar}_{\sigmazero,x} \big)(\xi) 
%% & = \phi^{\hbar}_{\sigmazero,x}(\xi) - {\HE(\mp q p) \over 1 \pm {2 i |p| \over m\, {B}(p)}}\,\big(\phi^{\hbar}_{\sigma,x}(\xi) - \phi^{\hbar}_{\sigma,x}(-\xi)\big) \\
& = \phi^{\hbar}_{\sigma,x}(\xi) - {\HE(\mp q p)\, \over 1 \mp {i\,\hbar^3 \over m \beta |p|}}\, \big(\phi^{\hbar}_{\sigma,x}(\xi) - \phi^{\hbar}_{\sigma,x}(-\xi)\big) \\
%% & = \phi^{\hbar}_{\sigma,x}(\xi) - \HE(\mp q p)\, (\pm \sgn(p) R_{\pm}(p/\hbar))\, \big(\phi^{\hbar}_{\sigma,x}(\xi) - \phi^{\hbar}_{\sigma,x}(-\xi)\big) \\
& = \phi^{\hbar}_{\sigma,x}(\xi) \mp \sgn(p)\, \HE(\mp q p)\,  R_{\pm}(p/\hbar)\, \big(\phi^{\hbar}_{\sigma,x}(\xi) - \phi^{\hbar}_{\sigma,x}(-\xi)\big) \\
& = \phi^{\hbar}_{\sigma,x}(q,p) + \sgn(q)\, \HE(\mp q p)\,  R_{\pm}(p/\hbar)\, \big(\phi^{\hbar}_{\sigma,x}(\xi) - \phi^{\hbar}_{\sigma,x}(-\xi)\big) \,.
\end{align*}
On the other hand, recalling the Identity \eqref{Omega_beta} established in Lemma \ref{lemma: Ompm}, noting the basic inequality $\big|R_{\pm}(k)\big| \leq 1$ and using the estimates reported in Lemmata \ref{LemmaEpbis}, \ref{lemma: sgnpsi} and \ref{LemmaE3}, we obtain the following for any $\eta \in (0,1)$:
\begin{align*}
& \left\|\Omega_\beta^{\pm}\, \psi^{\hbar}_{\sigmazero,\xi} - \Big[\psi^{\hbar}_{\sigmazero,\xi} + \sgn(q)\, \HE(\mp q p)\,  R_{\pm}(p/\hbar)\, \big(\psi^{\hbar}_{\sigmazero,\xi} - \psi^{\hbar}_{\sigmazero,-\xi}\big) \Big] \right\|_{L^2(\RE)}\\
& \leq C \left[ {\eta \over (1- \eta)}\, \Big({\hbar^3 \over m |\beta p|}\Big) + \Big(1 + {\hbar^3 \over m |\beta p|}\Big)\, e^{-\,\eta^2 {\sigmazero^2 p^2 \over 2\hbar}}	+ e^{- \frac{q^2}{4\hbar \sigmazero^2}} + e^{- {\sigmazero^2 p^2 \over \hbar}} \right] .
\end{align*}
Then, the proof is concluded recalling that $\psi^{\hbar}_{\sigma_0,\xi}(x) = \phi^{\hbar}_{\sigma,x}(\xi)$ (see Eq. \eqref{phidef}).
%%%%%%% SCATT

Next we prove claim \eqref{t2_2}.
To begin with, we apply the classical scattering operator with ${B}(p) := -(2  \beta/\hbar^3)\,p^2$ to the state $\phi^{\hbar}_{\sigmazero,(\cdot)}(\xi)$. Recalling the definition \eqref{eq: Rpm} of $R_{\pm}(k)$, from Eq. \eqref{eq: scattopcl} we obtain
\begin{align}
\big(S^{cl}_{B} \phi^{\hbar}_{\sigmazero,(\cdot)}\big)(\xi) =\, &\,  \phi^{\hbar}_{\sigmazero,(\cdot)}(\xi) + \sgn(p)\, {R_{-}(p/\hbar)} \big(\phi^{\hbar}_{\sigmazero,(\cdot)}(\xi) - \phi^{\hbar}_{\sigmazero,(\cdot)}(-\xi)\big)  \nonumber \\ 
=\, &\,  \psi^{\hbar}_{\sigmazero,\xi} + \sgn(p)\, R_{-}(p/\hbar)\,\big(\psi^{\hbar}_{\sigmazero,\xi} - \psi^{\hbar}_{\sigmazero,-\xi}\big)\,. \label{S-phi}
\end{align}
On the other hand, recalling the basic Identity \eqref{orto}, by simple addition and subtraction arguments and by the triangular inequality we get
\begin{equation*}
\begin{aligned}%\label{SEq1}
 \big\|(\Omega_\beta^{+})^{*} \Omega_\beta^-\, \psih_{\sigmazero,\xi} - \big(S^{cl}_{B} \phi^{\hbar}_{\sigmazero,(\cdot)}\big)(\xi) \big\|_{L^2(\RE)}   
\leq\, &  \,\left\|(\Omega_\beta^{+})^{*} \Big(\Omega_\beta^-\, \psih_{\sigmazero,\xi} - \big(W_{{B}}^- \phih_{\sigmazero,(\cdot)}\big)(\xi) \Big)\right\|_{L^2(\RE)} \\ 
& + 
\big\|(\Omega_\beta^{+})^{*} P_{ac} \big(W_{{B}}^-\, \phih_{\sigmazero,(\cdot)}\big)(\xi) - P_{ac} \big(S^{cl}_{B} \phi^{\hbar}_{\sigmazero,(\cdot)}\big)(\xi) \big\|_{L^2(\RE)}
 \\
&+ \big\|(\Omega_\beta^{+})^{*} P_{\beta} \big(W_{{B}}^- \,\phih_{\sigmazero,(\cdot)}\big)(\xi) \big\|_{L^2(\RE)}
+ \big\|P_{\beta} \big(S^{cl}_{B} \phi^{\hbar}_{\sigmazero,(\cdot)}\big)(\xi) \big\|_{L^2(\RE)} \,. 
\end{aligned}
\end{equation*}

Firstly, from Identity \eqref{S-phi} and Lemma \ref{LemmaProj} (see, in particular, Eq. \eqref{Pacsi0}), noting once more the basic inequality $|R_{\pm}(p/\hbar)| \leq 1$ we infer
\begin{align}
\big\|P_{\beta}\big(S^{cl}_{B} \phi^{\hbar}_{\sigmazero,(\cdot)}\big)(\xi) \big\|_{L^2(\RE)}
& \leq \big\|P_{\beta} \psih_{\sigmazero,\xi}\big\|_{L^2(\RE)} + \big|R_{-}(p/\hbar)\big|\, \Big(\big\|P_{\beta} \psih_{\sigmazero,\xi}\big\|_{L^2(\RE)} + \big\|P_{\beta} \psih_{\sigmazero,-\xi}\big\|_{L^2(\RE)} \Big) \nonumber \\
& \leq C \left({\hbar^{5} \sigma_0^2 \over m^2 \beta^2}\right)^{\!\!1/4} e^{{\hbar^5 \sigma_0^2 \over m^2 \beta^2}} \left( e^{- {\sigmazero^2 p^2 \over \hbar}} + e^{- \frac{q^2}{4\hbar \sigmazero^2}} \right) . \label{cortez1}
\end{align}
Secondly, let us notice that $\|\Omega_{\beta}^+ \psi\|_{L^2(\RE)} \leq \|\psi\|_{L^2(\RE)}$,  since $\Omega_{\beta}^+$  is the strong limit of operators with unit norm; thus, the same holds true for the adjoint $(\Omega_{\beta}^{+})^{*}$. Hence, by arguments similar to those described above, in view of Eq. \eqref{eq: waveopcl} and of Lemma \ref{LemmaProj}, we have
\begin{align}
\big\|(\Omega_{\beta}^{+})^{*} P_{\beta} \big(W_{{B}}^- \phih_{\sigmazero,(\cdot)}\big)(\xi) \big\|_{L^2(\RE)} 
& \leq \big\|P_{\beta} \big(W_{{B}}^- \phih_{\sigmazero,(\cdot)}\big)(\xi)\big\|_{L^2(\RE)} \nonumber \\
& \leq \big\|P_{\beta} \psih_{\sigmazero,\xi } \big\|_{L^2(\RE)}
+ \HE(q p)\,\big|R_{-}(p/\hbar)\big|\, \Big( \big\|P_{\beta}\psih_{\sigmazero ,\xi }\big\|_{L^2(\RE)}
+ \big\|P_{\beta}\psih_{\sigmazero,-\xi }\big\|_{L^2(\RE)}\Big)  \nonumber \\
& \leq C \left({\hbar^{5} \sigma_0^2 \over m^2 \beta^2}\right)^{\!\!1/4} e^{{\hbar^5 \sigma_0^2 \over m^2 \beta^2}} \left( e^{- {\sigmazero^2 p^2 \over \hbar}} + e^{- \frac{q^2}{4\hbar \sigmazero^2}} \right) . \label{cortez2}
\end{align}
To say more, again from the bound on $(\Omega_{\beta}^{+})^{*}$, we infer 
\[
\Big\|(\Omega_{\beta}^{+})^{*} \Big(\Omega_{\beta}^-\,\psih_{\sigmazero,\xi} - \big(W_{{B}}^-\, \phih_{\sigmazero,(\cdot)}\big)(\xi) \Big)\Big\|_{L^2(\RE)} 
\leq  \big\|\Omega_{\beta}^-\,\psih_{\sigmazero,\xi} - \big(W_{{B}}^- \phih_{\sigmazero,(\cdot)}\big)(\xi) \big\|_{L^2(\RE)} \,, 
\]
which is bounded by Eq. \eqref{t2_1} (proven previously). 
\\
Finally, on account of the unitarity of $\Omega_{\beta}^{+}$ on $\ran(P_{ac})$, we obtain
\[
\begin{aligned}
&  \big\|(\Omega_{\beta}^{+})^{*} P_{ac}\big(W_{{B}}^- \phih_{\sigmazero,(\cdot)}\big)(\xi) - P_{ac}\big(S^{cl}_{B} \phi^{\hbar}_{\sigmazero,(\cdot)}\big)(\xi) \big\|_{L^2(\RE)}  \\ 
&  = \big\| P_{ac} \big(W_{{B}}^- \phih_{\sigmazero,(\cdot)}\big)(\xi) - \Omega_{\beta}^{+}\, P_{ac} \big(S^{cl}_{B} \phi^{\hbar}_{\sigmazero,(\cdot)}\big)(\xi) \big\|_{L^2(\RE)} \\
&  \leq  \big\| \big(W_{{B}}^- \phih_{\sigmazero,(\cdot)}\big)(\xi) - \Omega_{\beta}^{+} \big(S^{cl}_{B} \phi^{\hbar}_{\sigmazero,(\cdot)}\big)(\xi) \big\|_{L^2(\RE)}  +  \big\| P_{\beta}  \big(W_{{B}}^- \phih_{\sigmazero,(\cdot)}\big)(\xi) \big\|_{L^2(\RE)}  +  \big\|  \Omega_{\beta}^{+} P_{\beta} \big(S^{cl}_{B} \phi^{\hbar}_{\sigmazero,(\cdot)}\big)(\xi) \big\|_{L^2(\RE)} \,.
\end{aligned}
\]
On the one hand, using once more arguments analogous to those described in the proof of the bounds \eqref{cortez1} and \eqref{cortez2}, we get
\[
\big\| P_{\beta} \big(W_{{B}}^- \phih_{\sigmazero,(\cdot)}\big)(\xi) \big\|_{L^2(\RE)}  \leq 
C \left({\hbar^{5} \sigma_0^2 \over m^2 \beta^2}\right)^{\!\!1/4} e^{{\hbar^5 \sigma_0^2 \over m^2 \beta^2}} \left( e^{- {\sigmazero^2 p^2 \over \hbar}} + e^{- \frac{q^2}{4\hbar \sigmazero^2}} \right)
\]
and
\[
\big\|  \Omega_{\beta}^{+} P_{\beta} \big(S^{cl}_{B} \phi^{\hbar}_{\sigmazero,(\cdot)}\big)(\xi) \big\|_{L^2(\RE)}  \leq 
\big\| P_{\beta}  \big(S^{cl}_{B} \phi^{\hbar}_{\sigmazero,(\cdot)}\big)(\xi) \big\|_{L^2(\RE)}  \leq 
 C \left({\hbar^{5} \sigma_0^2 \over m^2 \beta^2}\right)^{\!\!1/4} e^{{\hbar^5 \sigma_0^2 \over m^2 \beta^2}} \left( e^{- {\sigmazero^2 p^2 \over \hbar}} + e^{- \frac{q^2}{4\hbar \sigmazero^2}} \right) .
\]
On the other hand, since $ \Omega_{\beta}^{+} S^{cl}_{B} \phi^{\hbar}_{\sigmazero,(\cdot)}(\xi)  =  S^{cl}_{B}  \Omega_{\beta}^{+} \phi^{\hbar}_{\sigmazero,(\cdot)}(\xi)$ (due to the fact that the operators $ \Omega_{\beta}^{+}$ and $S^{cl}_{B}$ act on different variables), on account of Remark \ref{r:Wpm-unitary} we infer 
\[\begin{aligned}
&  \big\| \big(W_{{B}}^- \phih_{\sigmazero,(\cdot)}\big)(\xi) - \Omega_{\beta}^{+}  \big(S^{cl}_{B} \phi^{\hbar}_{\sigmazero,(\cdot)}\big)(\xi) \big\|_{L^2(\RE)}   \\
& = \big\|\big(W_{{B}}^-\phih_{\sigmazero,(\cdot)}\big)(\xi)  - \big(S^{cl}_{B} \Omega_{\beta}^{+} \phi^{\hbar}_{\sigmazero,(\cdot)}\big)(\xi) \big\|_{L^2(\RE)}  \\ 
& =  \big\|\big(S^{cl}_{B} \big(\Omega_{\beta}^{+} - W_{{B}}^+\big)\phi^{\hbar}_{\sigmazero,(\cdot)}\big)(\xi) \big\|_{L^2(\RE)}  \\
& \leq C\left(\big\| \big(\big(\Omega_{\beta}^{+} - W_{{B}}^+\big)\phi^{\hbar}_{\sigmazero,(\cdot)}\big)(\xi) \big\|_{L^2(\RE)} + \big\| \big(\big(\Omega_{\beta}^{+} - W_{{B}}^+\big)\phi^{\hbar}_{\sigmazero,(\cdot)}\big)(-\xi) \big\|_{L^2(\RE)} \right)  \\ 
& \leq C \left[ {\eta \over (1- \eta)}\, \Big({\hbar^3 \over m |\beta p|}\Big) + \Big(1 + {\hbar^3 \over m |\beta p|}\Big)\, e^{-\,\eta^2 {\sigmazero^2 p^2 \over 2\hbar}}	+ e^{- \frac{q^2}{4\hbar \sigmazero^2}} + e^{- {\sigmazero^2 p^2 \over \hbar}} \right]\,,
 \end{aligned}
 \]
where in  the last two inequalities we used $ |\big(S^{cl}_{B} f\big)(q,p) | \leq C \big(|f(q,p)|+|f(-q,-p)|\big)$ (see Eq. \eqref{eq: scattopcl} and note that $\big|1 - {2i\,|p| \over m\,{B}(p)}\big|^{-1}\leq 1$) and the bound in Eq. \eqref{t2_1}. 

Summing up, the above estimates imply Eq. \eqref{t2_2}. 
%\end{proof}

The proof of Corollary \ref{cor: Ompm} is similar to that of Corollary \ref{c:1} and we omit it for brevity.

\end{document}